\documentclass{article}

\usepackage{fullpage}
\usepackage{amsmath,amsthm,amssymb,amsfonts}

\newcommand{\citep}[1]{\cite{#1}}
\newcommand{\citet}[1]{\cite{#1}}

\usepackage{bbm}  

\usepackage{thmtools}
\usepackage{thm-restate}

\usepackage{tikz}
\usetikzlibrary{calc}

\usepackage{caption}
\usepackage{subcaption}

\usepackage{authblk}

\usepackage{xcolor}
\usepackage[hidelinks,colorlinks=true,linkcolor=blue]{hyperref}
\hypersetup{
	hidelinks,  		
	colorlinks=true,	
	citecolor=blue, 	
	linkcolor=blue
}
\usepackage{cleveref}

\newif\ifshowcomments
\showcommentsfalse   
\newcommand{\docomment}[3]{{\ifshowcomments \textcolor{#1}{[ #2 : #3 ]} \fi}}
\newcommand{\bo}[1]{\docomment{brown}{Bo}{#1}}
\newcommand{\robin}[1]{\docomment{blue}{Robin}{#1}}
\newcommand{\elias}[1]{\docomment{purple}{Elias}{#1}}

\newtheorem{theorem}{Theorem}[section]
\newtheorem*{theorem*}{Theorem}
\newtheorem{lemma}{Lemma}[section]
\newtheorem*{lemma*}{Lemma}
\newtheorem{proposition}{Proposition}[section]
\newtheorem*{prop*}{Proposition}

\newtheorem{corollary}{Corollary}[section]
\newtheorem*{corollary*}{Corollary}
\newtheorem{observation}{Observation}[section]

\theoremstyle{definition}\newtheorem{definition}{Definition}[section]
\theoremstyle{definition}\newtheorem*{definition*}{Definition}
\theoremstyle{definition}\newtheorem{reduction}{Reduction}
\theoremstyle{definition}\newtheorem{algorithm}{Algorithm}[section]

\DeclareMathOperator*{\E}{\mathbb{E}}
\DeclareMathOperator*{\argmax}{argmax}

\newcommand{\reals}{\mathbb{R}}
\newcommand{\Indic}[1]{\mathbbm{1}\mathopen{}\left[#1\right]}

\newcommand{\A}{\mathbb{A}}

\newcommand{\F}{\mathcal{F}}        
\newcommand{\PF}{\mathcal{P}_{\F}}  
\newcommand{\D}{\mathcal{D}}		

\newcommand{\M}{\mathcal{M}}        
\newcommand{\C}{\mathcal{C}}

\newcommand{\ALG}{\mathrm{ALG}}
\newcommand{\OPT}{\mathrm{OPT}}
\newcommand{\OPTHAT}{\mathrm{\widehat{OPT}}}

\newcommand{\Perf}{\mathrm{Perf}}			
			
\newcommand{\Perfofon}[2]{\Perf^{#1}(#2)}	
\newcommand{\Welf}{\mathrm{Welf}}			
			
\newcommand{\Welfofon}[2]{\Welf^{#1}(#2)}

\newcommand{\SAUP}{\textsc{SAUP}}

\newcommand{\I}{\mathcal{I}} 
\newcommand{\bI}{\mathbb{I}} 

\newcommand{\cms}{\textsc{CMS}}
\newcommand{\msp}{\textsc{MSP}}
\newcommand{\cob}{\textsc{COB}}
\newcommand{\cobs}{\cob{}\ensuremath{\mathrm{~Selection}}}
\newcommand{\cor}{\textsc{COR}}
\newcommand{\cors}{\cor{}\ensuremath{\mathrm{~Selection}}}

\newcommand{\vars}{\ensuremath{\mathrm{Variable~Selection}}}

\newcommand{\MSPtoCOB}{\ensuremath{\mathrm{MSPtoCOB}}}
\newcommand{\COBtoCOR}{\ensuremath{\mathrm{COBtoCOR}}}

\newcommand{\pnoi}{\textsc{PNOI}}
\renewcommand{\mp}{\textsc{matroid-prophet}}

\newcommand{\Val}{\ensuremath{\mathrm{Val}}}

\begin{document}

\begin{titlepage}

\title{Combinatorial Markov Search}

\author{Robin Bowers}
\author{Elias Lindgren}
\author{Bo Waggoner}
\affil{University of Colorado, Boulder}
\date{}

\maketitle

\robin{make sure comments are off!}

A decisionmaker faces $n$ alternatives, each of which represents a potential reward. 
After investing costly resources into investigating the alternatives, the decisionmaker selects one (or more generally a feasible subset), and receives the associated reward(s). 
We model each alternative as a $\textit{Markov Search Process}$ (MSP), a type of undiscounted Markov Decision Process on a finite acyclic graph, and call this problem $\textit{Combinatorial Markov Search}$ (CMS). 
CMS broadly generalizes recent NP-hard problems of interest such as Pandora's Box with nonobligatory inspection. 
Despite the seemingly adaptive and interactive nature of the problem, we construct online algorithms for CMS that explore each alternative sequentially, either selecting or discarding it before moving to the next.
We first show that any ex-ante prophet inequality can be converted into an (inefficient) online algorithm for CMS with the same approximation guarantee.
Then, for any matroid feasibility constraint, we construct a polynomial-time $(1/2-\epsilon)$-approximation algorithm for CMS. 
Our construction also implies incentive-compatible mechanisms with constant Price of Anarchy for a strategic version of the problem that generalizes auctions with inspection costs.

\end{titlepage}

\section{Introduction} \label{sec:intro}

In much recent work growing out of the Pandora's Box literature~\citep{beyhaghi2024recent}, a decisionmaker is presented with a set of alternatives of uncertain value.
She must dynamically acquire costly information and choose a valuable alternative or subset of alternatives.
We propose a general model for this problem in which the decisionmaker's $n$ alternatives $\M_1,\dots,\M_n$ are each a \emph{Markov Search Process ($\msp$)}, a variant of a Markov Decision Process.\footnote{A $\msp$ has a finite directed acyclic graph structure; no discount factor; actions only generate costs (not rewards); and has terminal states that are associated with an available reward.}
Each $\msp$ $\M_i$ is defined by a choice of actions, each associated with a cost and a stochastic transition to a new state; the new state has a new choice of costly actions, and so on until a terminal state is reached, revealing an available reward for choosing alternative $\M_i$.

In the \emph{Combinatorial Markov Search ($\cms$)} problem, the decisionmaker is given a full description of $\M_1,\dots,\M_n$.
She interacts by repeatedly selecting any $\M_i$, taking an action, paying the incurred cost, and observing its state transition.
She then selects another $\M_i$ (or possibly the same one), and so on.
At any point, she may decide to stop and claim some subset of the $\msp$s, obtaining their associated revealed rewards. 
The accepted subset $F$  must belong to a given feasibility system $\F \subseteq 2^{\{1,\dots,n\}}$, such as a matroid, and we call the associated problem $\cms(\F)$.
For example, in \emph{single selection}, $\F = \{\emptyset, \{1\},\dots,\{n\}\}$, denoting that at most one of the $\msp$s can be claimed.
The \emph{welfare} of the algorithm is the sum of rewards claimed minus the sum of all costs incurred.
The algorithm is an $\alpha$-approximation if its expected welfare on every instance is at least an $\alpha$ fraction of the expected welfare of the optimal, fully-adaptive, computationally unconstrained algorithm.

$\cms$ models problems such as a company investing in the development of its next product, which may be one of $n$ potential options.
Each option is a Markov Search Process involving research and development choices that require resources and may open the door to further choices.
The company may interactively investigate and develop the potential options in parallel, e.g. abandoning some when they become less promising, and eventually choosing one to bring to market.

\elias{this para edited}
It is known that approximation algorithms for combinatorial optimization can often be extended to Pandora's Box models and related tractable models of information acquisition such as multi-stage Pandora's Box~\citep{guha2007information,singla2018price,gupta2019markovian,esfandiari2019online}.
Recent work has begun to consider \emph{intractable} models, such as \emph{Pandora's Box with non-obligatory inspection}~\citep{doval2018whether,fu2023pandora,beyhaghi2023pandora}, which is NP-hard even for single selection.
Our work vastly generalizes such models (see Appendix \ref{app:hardness}), and we show that any problem that can be represented within our $\cms$ framework admits an efficient approximation algorithm (Theorem \ref{thm:main-2}).
Previously, efficient approximations were generally known only for very specific special cases of $\cms$.
\elias{such as nonobligatory? TODO what other known problems do we capture, with or without previous approximations?
previously also said ``no negative results on approximability have been shown'', but it's clear why since we prove approximability for all of them}
Variants of this question have been studied both classically~\citep{nash1973optimal,glazebrook1982sufficient,whittle1980multi} and in recent~\citep{ma2018improvements} and independent~\citep{chawla2025combinatorial} work (see Section \ref{sec:related}).
\elias{anything to say about what we contribute to these variants? Not a lot for settings w/o stopping time (classical results we cite here)}

\subsection{Main Results}

Combinatorial Markov Search appears to require highly interactive, adaptive algorithms.
We first ask whether this is truly the case by considering a restricted class of \emph{online} algorithms\footnote{We emphasize that, as with prophet inequalities, our online algorithms are given the full problem description in advance.} as approximations to the optimal offline algorithm. 
These algorithms must fully interact with $\M_i$ and ``take it or leave it'' before beginning any interactions with $\M_{i+1}$. 
In the special case of our problem where each $\M_i$ has no actions at all and simply reveals a stochastic reward, existence of an online algorithm with approximation guarantee $\alpha$ is known as an \emph{$\alpha$ prophet inequality}.

Our first main result is that a prophet inequality for feasibility constraint $\F$ can be amplified into an online algorithm for $\cms(\F)$ with the same approximation guarantee.

\begin{theorem}[Online Approximation] \label{thm:main-1}
  Let $\F$ be a nonempty downward-closed set system.
  If there is an ex-ante\footnote{An approximation guarantee is \emph{ex-ante} if the benchmark is a fully-adaptive optimal algorithm that is only required to satisfy the feasibility constraints in expectation. The approximation algorithm is still required to always be feasible. See Section \ref{sec:prelims}.} $\alpha$ prophet inequality for $\F$, then there is an online $\alpha$-approximation algorithm, not necessarily computationally efficient, for Combinatorial Markov Search$(\F)$.
\end{theorem}
For single selection, the optimal prophet inequality is $\tfrac{1}{2}$~\citep{krengel1978semiamarts,samuelcahn1984comparison}, meaning that the ``fully-adaptive'' algorithm that simply picks the largest of $n$ rewards can perform at most twice as well as an online algorithm.
The $\tfrac{1}{2}$ guarantee extends to matroid feasibility constraints $\F$ even against the ex-ante benchmark~\citep{lee2018optimal}, so Theorem \ref{thm:main-1} implies that even for matroids, even when all alternatives may be highly dynamic Markov Search Processes involving sequences of costly decisions and stochastic transitions, the benefit of sharing information across $\msp$s is no higher: the benefit of full adaptivity is still only a factor of two.
Theorem \ref{thm:main-1} is similar in spirit to an \emph{adaptivity gap} in stochastic probing problems (e.g. \citet{gupta2016algorithms}).

The approximation algorithms for $\cms$ in Theorem \ref{thm:main-1} are not computationally efficient. 
By substituting our own ``matroid prophet inequality'' and designing additional efficient subroutines, we obtain our second main result.
\begin{theorem}[Efficient Approximation] \label{thm:main-2}
  Let $\F$ be a matroid.
  There exists an algorithm, running in time polynomial in the input size and $\frac{1}{\epsilon}$, that for any $\epsilon > 0$ achieves a $(\tfrac{1}{2}-\epsilon)$-approximation for Combinatorial Markov Search$(\F)$.
  Moreover, the algorithm is online.
\end{theorem}
\elias{edits in this para}
We thus show that there exists an efficient constant-factor approximation algorithm for $\cms$, not only for single selection but also for any matroid constraint.
\elias{TODO if we know of any special cases for which it provides/improves approximation factor, mention here.
Could mention nonobligatory but seems insufficient on its own especially as we don't improve approx factor.}
Furthermore, the efficient algorithm is ``non-adaptive'' in that it interacts with the MSPs one at a time online.

\paragraph{Mechanism Interpretation.}
It is well-known that prophet inequalities are \emph{incentive-compatible (IC)}: they can be converted to mechanisms for e.g. selling items to arriving buyers~\citep{lucier2017economic}.
The thresholds used by prophets algorithms are interpreted as posted prices, to which the buyers best-respond, implementing the algorithm.

We define the $\cms(\F)$-Auction problem where an $\F$-feasible subset of buyers may be selected as winners, and each buyer interacts privately with their own $\msp$ to discover their value for winning.
Notably, our generic reduction in Theorem \ref{thm:main-1}, despite beginning with an IC prophet inequality, converts it to an algorithm that is \emph{not} IC. 
However, an essential component of our efficient algorithm for Theorem \ref{thm:main-2} is incentive compatibility.
The algorithm behaves on each $\msp$ as if it were an agent attempting to maximize the profit from that $\msp$ alone subject to a posted price.
Converting the algorithm to a mechanism yields the following result.

\begin{corollary} \label{cor:main-poa}
  If $\F$ is a matroid, there exists a mechanism for the $\cms(\F)$-Auction problem with Price of Anarchy $\tfrac{1}{2}$.
  Furthermore, for any $\epsilon > 0$, there exists a polynomial-time mechanism for the $\cms(\F)$-Auction problem with Price of Anarchy $\tfrac{1}{2}-\epsilon$.
\end{corollary}
In subsequent work, \citet{chawla2025commitment} use a \emph{commitment gap} approach to improve the guarantee to $1-1/e$.

\subsection{Techniques}
In this section, we provide a high-level roadmap of the primary techniques used in this work.

\paragraph{Ex-Ante Relaxation.}
Recall that the set of MSPs we may select is constrained, typically according to a matroid constraint.
We compare our \emph{ex-post} feasible algorithm to the optimal \emph{ex-ante} feasible algorithm (Definition \ref{def:ex-ante}) as a benchmark.
The ex-ante benchmark may interact with each MSP independently, subject to some fixed probability of claiming.
As any ex-post feasible selection is also ex-ante feasible, this provides an upper bound on the value of any optimal algorithm for ex-post CMS.

\paragraph{Online Algorithms.}
The core technique of this work is leveraging constant-factor approximations given by prophet inequalities to obtain the same approximation factor for CMS.
In the classic prophet inequality setting, a decisionmaker must select items from a set of many alternatives of uncertain value.
Simple thresholding techniques allow us to approximate the optimal solution when items must be considered one at a time~\citep{samuelcahn1984comparison}.
These techniques are robust to replacing the random variables that the decisionmaker chooses between with Pandora's Boxes~\citep{kleinberg2016descending}.
We show in Section \ref{sec:inefficient} that these techniques extend even when we replace these alternatives with MSPs, which are much more complex objects.
We do this via a chain of reductions from CMS to two other novel settings, which we call \emph{choice-over-bandits} and \emph{choice-over-rewards}, and finally to the original prophet inequality setting.

In Section \ref{sec:efficient} we show how to overcome inefficiencies in these reductions, and arrive at an efficient constant-factor approximation (Theorem \ref{thm:main-2}).
\elias{Originally described prophet inequality as single selection and then mentioned work extending to matroid constraints, removed per Robin's suggestion}

\paragraph{Choice-Over-Bandits.}
We reduce CMS to the \emph{choice-over-bandits} (COB) setting in Section \ref{subsec:msp-to-cob}.
A \emph{bandit process} is a MSP in which there is only one action available at each step.
If we commit to following a particular policy when interacting with a given MSP, we are constrained to a single action in each state (or a fixed distribution over actions, if the policy is randomized).
We can therefore view a MSP with a fixed policy as a bandit process.
If we are presented with a choice over many policies for a given MSP, we may view the MSP as a ``choice-over-bandits,'' defined formally in Definition \ref{def:cor}.
Bandit processes are well-studied, and an efficient optimal algorithm for interacting with and selecting a matroid subset of bandits is known~\citep{kleinberg2016descending,singla2018price,bowers2024matching}, 
Unfortunately, the reduction from a MSP to a COB is very inefficient, as there are exponentially many possible policies for any MSP, and therefore exponentially many bandits to choose from.
However, we show in Section \ref{subsec:threshold} how to efficiently choose the ``best'' policy (with respect to some threshold) for any MSP, and accordingly choose the ``best'' bandit from the corresponding COB, without computing the entire reduction.
We call this problem of optimally exploring a MSP subject to a given threshold the \emph{single agent utility problem} or ``SAUP'' for short, defined formally in Definition \ref{def:saup}.

\paragraph{Choice-Over-Rewards.}
We further reduce choice-over-bandits to \emph{choice-over-rewards} (COR) in Section \ref{subsec:cob-to-cor}.
A \emph{choice-over-rewards} process (Definition \ref{def:cor}) is a choice-over-bandits in which, after the initial choice, we immediately observe a reward.
We reduce COB to COR using well-known machinery from the Pandora's Box literature~\citep{weitzman1979optimal,kleinberg2016descending,singla2018price,bowers2024matching} to amortize the inspection costs of each bandit into a random variable.
This reduction is the most straightforward of the three and the only one that is immediately efficient.

In Section \ref{subsec:prophets}, we finally reduce the problem of choice-over-rewards selection to the classic prophet inequality problem, concluding our chain of reductions.
Using the structure of the ex-ante optimal solution to the COR problem, we collapse all choices into a single random variable reward, drawn from a distribution over the values of each individual possible reward according to the (randomized) choice of reward made by the ex-ante optimal algorithm.

We provide a direct prophet inequality-style approximation algorithm for the COR model in Section \ref{subsec:matroid-prophet}, and combine this alternate prophet algorithm with our SAUP subroutine to produce an efficient approximation algorithm for CMS with matroid feasibility structures in Section \ref{sec:efficient}.
\elias{copied Robin's edits in verbatim}

\elias{Mention mechanism design perspective? Conceptually important, technically trivial}

\section{Notation And Preliminaries} \label{sec:prelims}
\subsection{The Combinatorial Markov Search (\cms{}) Problem}

Our main problem involves interacting with and selecting among $n$ Markov Search Processes (MSPs), each defined as follows.

\begin{definition}[Markov Search Process] \label{def:msp}
A \emph{Markov Search Process} is a tuple $\M=(S,A,P,C,V)$.
$S$ is a finite set of states and $A$ a finite set of actions.
$P:S \times A \times S \to [0,1]$ is the \emph{transition function}, where $P(s';a,s)$ represents the probability of transitioning to state $s'$ given that the process is currently in state $s$ and action $a$ is taken.
There is a set of \emph{legal actions} $A^s \subseteq A$ for each state $s$.
For all $a \in A^s$, $\sum_{s'\in S} P(s';a,s)=1$.
$C:A\times S \to \reals_{\geq 0}$ is the \emph{cost function}.
$C(a,s)$ is the cost incurred by taking action $a \in A^s$ when the current state is $s$.
Finally, $V:S\to\reals_{\geq 0}$ is the \emph{reward function}, representing an available reward at state $s$.

Given a MSP $\M$, the \emph{state graph} of $\M$ contains a directed edge $(s,s')$ if $P(s';a,s) > 0$ for some $a$.
We assume the state graph is a directed acyclic graph (DAG).
The \emph{start state} of $\M$ is the unique state $s^*$ with no incoming edges.
A \emph{terminal state} of $\M$ is a state $s$ with no outgoing edges, i.e. $A^s = \emptyset$.
We assume that rewards are only positive on terminal states, i.e., $V(s)=0$ unless $s$ is a terminal state.
\end{definition}

\begin{figure}
	\begin{center}
		\begin{tikzpicture}[
			wide/.style={line width=4pt}, every node/.style={circle,draw,minimum size=12}, scale=.2]

			\node (root) at (-8,0) {};

			\coordinate (branch1) at (-4, 3) {};
			\coordinate (branch2) at (-4, -3) {};

			\node  (s1) at (0, 4.5) {};
			\node  (s2) at (0, 1.5) {};
			\node  (s3) at (0, -1.5) {};
			\node  (s4) at (0, -4.5) {};

			\coordinate (branch3) at (4, 6) {};
			\coordinate (branch4) at (4, 3) {};
			\coordinate (branch5) at (4, 0) {};
			\coordinate (branch6) at (4, -3) {};
			\coordinate (branch7) at (4,-6) {};

			\node  (s5) at (8, 6) {};
			\node  (s6) at (8, 3) {};
			\node  (s7) at (8, 0) {};
			\node  (s8) at (8, -3) {};
			\node  (s9) at (8,-6) {};

			\draw (root) -- (branch1);
			\draw (root) -- (branch2);

			\draw[dashed] (branch1) -- (s1);
			\draw[dashed]  (branch1) -- (s2);
			\draw[dashed]  (branch1) -- (s3);

			\draw[dashed] (branch2) -- (s2);
			\draw[dashed]  (branch2) -- (s3);
			\draw[dashed]  (branch2) -- (s4);

			\draw (s1) -- (branch3);
			\draw (s1) -- (branch4);

			\draw (s3) -- (branch5);
			\draw (s3) -- (branch6);
			\draw (s3) -- (branch7);

			\draw[dashed] (branch3) -- (s5);

			\draw[dashed] (branch4) -- (s6);
			\draw[dashed]  (branch4) -- (s7);

			\draw[dashed] (branch5) -- (s5);

			\draw[dashed] (branch6) -- (s7);
			\draw[dashed] (branch6) -- (s8);

			\draw[dashed] (branch7) -- (s9);
			\draw[dashed]  (branch7) -- (s6);
			\draw[dashed]  (branch7) -- (s8);

		\end{tikzpicture}

		\caption{
			An example of a general Markov Search Process (MSP).
			States are given by nodes in the graph.
			Solid lines represent available actions, while dotted lines represent possible randomized state transitions given the chosen action.
		}\label{fig:msp-random}
	\end{center}
\end{figure}

\begin{definition}[Combinatorial Markov Search$(\F)$] \label{def:cms}
Let $\F \subseteq 2^{[n]}$ be downward-closed.\footnote{We write $[n]$ for $\{1,\dots,n\}$. $\F$ is \emph{downward-closed} if $F \in \F, F' \subseteq F \implies F' \in \F$.}
An instance $\I$ of the \emph{Combinatorial Markov Search$(\F)$} or \emph{\cms$(\F)$} problem consists of a set of $n$ independent MSPs $\M_1,\ldots,\M_n$.
Each $\M_i$ begins in its start $s^*_i$.
An algorithm begins with an empty solution set $F = \emptyset$.
At any time, the algorithm may either:
\begin{enumerate}
\item Choose any MSP $\M_i$ in a non-terminal state $s_i$ and take some action $a_i\in A_i^{s_i}$.
      The algorithm incurs cost $C_i(a_i,s_i)$ and $\M_i$ transitions to a new state $s'_i$ drawn independently according to $P_i(~ \cdot ~ ;a_i,s_i)$.
\item \emph{Claim} a MSP $\M_i$ in a terminal state $s_i$, meaning to add $i$ to $F$.
      The algorithm receives reward $V_i(s_i)$.
\item Halt.
\end{enumerate}
We refer to either (1) or (2) as \emph{advancing} $\M_i$.
The algorithm is \emph{ex-post feasible} if $F\in\F$ with probability $1$.
Unless otherwise stated, all algorithms are assumed to be ex-post feasible.
\end{definition}

\paragraph{Welfare and performance.}
Given an algorithm and an instance $\I = (\M_1,\dots,\M_n)$, with $\M_i = (S_i,A_i,P_i,C_i,V_i)$, we write $\A_{i,s}$ as the indicator variable for claiming $\M_i$ in state $s$ and $\bI_{i,a,s}$ as the indicator for advancing $\M_i$ from $s$ by taking action $a$.
The \emph{performance} of an algorithm on a particular MSP $\M_i$ is the net value, i.e. reward claimed (if any) minus the sum of all costs incurred.
\[\Perf(\M_i)=\sum_{s\in S_i} \A_{i,s} V_i(s) - \sum_{s\in S_i, a\in A_i} \bI_{i,a,s}C_i(a,s) . \]
The \emph{welfare} of the algorithm on the instance is
\[\Welf(\I) = \sum_{i=1}^n \Perf(\M_i) . \]
We sometimes write $\Perf^{\ALG}$ and $\Welf^{\ALG}$ to clarify reference to a particular algorithm $\ALG$.
We may similarly write $\A_{i,s}^{\ALG}$ and $\bI_{i,a,s}^{\ALG}$. 
We write $\A_i := \sum_{s \in S_i} \A_{i,s}$.

The objective in CMS is to maximize expected welfare.
An algorithm $\ALG$ \emph{is an $\alpha$-approximation} if, for all algorithms $\OPT$ and all instances $\I$, $\E[ \Welf^{\ALG}(\I) ] \geq \alpha \E[ \Welf^{\OPT}(\I) ]$, with expectation taken over randomness of the instance's transitions and the algorithms themselves.

Our algorithms, though given the entire problem description in advance, will interact \emph{online}.
\begin{definition}[Online, Fully-Adaptive]\label{def:online}
  An algorithm is \emph{online} if, for each $i=1,\dots,n$, the algorithm must complete all interactions with $\M_i$ and decide irrevocably whether to claim $\M{}_i$ or not before it can take any action in $\M_{i+1}$.
  An algorithm that is not required to be online is \emph{fully-adaptive}.
  In either case, the full description of each $\M_i$ is given. 
\end{definition}

\subsection{Ex-ante \cms{} and Policies}
We will compare our algorithms, which are always ex-post feasible, to a stronger benchmark: the optimal algorithm for \cms{} that is only required to produce a feasible set ``in expectation.''
\begin{definition}[$\PF$, ex-ante feasible] \label{def:ex-ante}
	Given $F \subseteq [n]$, let $1_F \in [0,1]^n$ denote the indicator vector for $F$, i.e. $1_F(i) = \Indic{i \in F}$.
	Given $\F \subseteq 2^{[n]}$, define $\PF$ to be the convex hull of $\{1_F : F \in \F\}$.
	An algorithm $\OPT$ is termed \emph{ex-ante feasible} if the vector $Q$ defined by $Q_i = \Pr[i \in \OPT]$ satisfies $Q \in \PF$.
    The probability is taken over both the algorithm's coins and the random transitions.
\end{definition}
We also refer to $\PF$ as the \emph{feasible polytope}. 
For example, in single selection, $\PF$ consists of nonnegative vectors whose entries sum to at most one.
In that case, an example of an ex-ante feasible algorithm that is not ex-post feasible is one that claims each $\M_i$ with probability $\tfrac{1}{n}$ independently.

We say an algorithm $\ALG$ \emph{is an ex-ante $\alpha$-approximation} if, for all ex-ante feasible algorithms $\OPT$ and all instances $\I$, $\E[ \Welf^{\ALG}(\I) ] \geq \alpha \E[ \Welf^{\OPT}(\I) ]$.
Because every ex-post feasible $\OPT$ is also ex-ante feasible, this implies that $\ALG$ simply guarantees an $\alpha$-approximation as well.

\paragraph{Policies and independence.}
We use the term \emph{policy} for a ``local algorithm'' that is defined on a single MSP $\M_i = (S_i,A_i,P_i,C_i,V_i)$.
Formally, a policy $\rho$ is a possibly-randomized function from a \emph{history} $h=(s^*,a_1,s_1,\ldots,a_t,s_t)$ of states and actions in $\M_i$ to either ``halt'', ``claim'', or a legal next action $\rho(h) = a \in A^{s_t}$.
When executing a policy $\rho$ on an instance $\M_i$, we define the indicators $\A_{i,s}^{\rho}$ and $\bI_{i,a,s}^{\rho}$ and the performance $\Perf^{\rho}(\M_i)$ exactly as for algorithms above.
We typically use $\pi$ to denote a deterministic policy and $\rho$ for a possibly-randomized one.
When a policy depends only on the current state $s$, we write $\rho(s)$ rather than $\rho(h)$.

The following observation is crucial to our techniques throughout: the ex-ante benchmark is in a sense ``non-adaptive'', as our online algorithms will be.
\begin{observation} \label{obs:ex-ante-indep}
  Any ex-ante feasible algorithm for \cms{}, without loss of generality, consists of $n$ policies $\rho_1,\dots,\rho_n$ that are executed independently on $\M_1,\dots,\M_n$, respectively.
\end{observation}
\begin{proof}
  Let $\ALG$ be any ex-ante feasible algorithm.
  For each $i$, define $\rho_i$ to be the following policy on $\M_i$: run $\ALG$ on $\M_i$ together with simulated independent copies of $(\M_{i'} : i' \neq i)$.
  The algorithm that executes each $\rho_i$ independently on $\M_i$ has the same probability of claiming each $\M_i$ as $\ALG$ does, so it is ex-ante feasible, and it also has the same expected welfare as $\ALG$.
\end{proof}
We will use the above observation to analyze, in particular, the behavior of any ex-ante optimal algorithm for \cms{}.

\subsection{Bandits and Weitzman Indices} \label{subsec:bandits}

It will be useful throughout to define a special case of a Markov Search Process that we call a bandit process.
As in the Bayesian Multi-Armed-Bandit setting, a bandit process has only one available action at every time.
The action can be interpreted as paying a cost to ``pull the arm'', observing a state transition, possibly repeating until an available reward is revealed.
\begin{definition}[Bandit Process] \label{def:bandit}
A MSP $\M = (S,A,P,C,V)$ is a \emph{bandit process} if $|A^s| \leq 1$ for all $s\in S$.
\end{definition}

Bandits admit the following analysis tools that make them highly compatible with greedy algorithms~\citep{kleinberg2016descending, singla2018price, lee2018optimal, gupta2019markovian}, a fact that will be crucial to our approach.
For instance, the most general case of \cms{} that is known to have an efficient optimal algorithm is the case where every $\M_i$ is a bandit and $\F$ is a matroid \citep{gupta2019markovian}.
The optimal algorithm in this case defines the following ``Weitzman index'' for each bandit and always advances the bandit whose Weitzman index is currently largest.\footnote{This model closely resembles the Bayesian Multi-Armed-Bandits problem, a discounted infinite-horizon version of our problem where there is no selection. In this problem, when each alternative is a bandit, the Gittins Index Theorem gives an optimal algorithm. However, if the alternatives are generalized from bandits to general Markov Decision Processes, no approximation algorithm is known \citep{nash1980generalized}.}
\begin{definition}[Weitzman Index, Capped Value (following \citet{bowers2024matching})] \label{def:weitzman}
Given a bandit process $B=(S,A,P,C,V)$, the \textit{Weitzman index} $\sigma_s$ and the \textit{capped value} $\kappa_s$ of each state $s$ are defined by backwards induction as follows.
For any terminal state $s$, $\sigma_s = \kappa_s = V(s)$.
For any non-terminal state $s$, let the random variable $n(s)$ be the next state reached when advancing from $s$. Define $\sigma_s$ to be the unique value satisfying
\[\E_{n(s)}[(\kappa_{n(s)}-\sigma_s)^+]=C(s)\]
where $(\: \cdot \:)^+ = \max\{\cdot,0\}$ and $C(s)$ is the cost of taking the unique ``advance'' action from $s$.
We define $\kappa_s := \min\{\sigma_s,\kappa_{n(s)}\}.$
We also write $\kappa = \kappa_{s^*}$ where $s^*$ is the initial state of the process.
\end{definition}
We emphasize that $\kappa_s$, $n(s)$ are random variables, while $\sigma_s$ is a constant which can be computed prior to the realization of any randomness. 

The Weitzman index and capped value are useful because of the following definition and lemma.
\begin{definition}[Exposure] \label{def:exposed}
A policy $\rho$ is \textit{exposed} on a bandit $B$ if, with nonzero probability, $\rho$ advances from some state $s$, eventually visits some state $s'$ with $\sigma_{s'} > \sigma_{s}$, and does not advance from $s'$.
Otherwise, $\rho$ is \textit{non-exposed}.
\end{definition}
\begin{lemma}[\citet{kleinberg2016descending,bowers2024matching}] \label{lem:magic}
For any bandit process $B$ and policy $\rho$, $\E[\Perf^{\rho}(B)] \leq \E[\A^{\rho}\kappa]$, with equality if and only if $\rho$ is non-exposed on $B$.
\end{lemma}
The idea is to compare performance of $\rho$ to an ``amortized world'' where policies never pay any costs to advance, but only receive the reward $\kappa$ instead of $V(s)$ when claiming.
Lemma \ref{lem:magic} states that $\rho$'s performance is always upper-bounded by the performance of the same policy in the amortized world.
Furthermore, this upper bound is an equality if $\rho$ satisfies the non-exposure condition, which ensures that transition costs are always amortized into rewards.

\robin{begin edits}

While the primary goal of this work is developing an understanding of general $\msp$s, we note that exposure is defined only in the restricted setting of bandits (Definition \ref{def:bandit}), and with regard to a single bandit rather than an entire $\cms{}$ instance. 

In general, optimal algorithms with non-exposed policies exist in the setting where each \msp{} is a bandit. 
Versions of this result are related to the frugal algorithms of \citet{singla2018price} and the descending procedures of \citet{bowers2024matching}, which largely consist of advancing the bandit currently in a state with the highest $\sigma$, and a more general result is given in \citet{gupta2019markovian}.  

Beyond bandit settings, the concept of non-exposure deteriorates as the expectation of future values becomes dependent on future actions. 
Appendix \ref{sec:weitzman-interpretation} considers an exposure and Weitzman index-based interpretation of the algorithms developed in Section \ref{sec:inefficient}.

\robin{end edits}

\section{Reducing Combinatorial Markov Search to Prophet Inequalities} \label{sec:inefficient}

In this section, we prove Theorem \ref{thm:main-1}: whenever the constraint system $\F$ admits an ex-ante prophet inequality, the \cms$(\F)$ problem admits an online algorithm with the same approximation factor, albeit inefficiently. 
A prophet inequality bounds the advantage of being able to select a subset of random rewards ``offline'' versus needing to select them online.
Theorem \ref{thm:main-1} states that even when rewards are replaced by arbitrary Markov Search Processes, and the offline algorithm may explore the MSPs concurrently, its advantage is no greater.

Our approach consists of a chain of three reductions to simpler versions of the \cms{} problem where the Markov Search Processes $\M_1,\dots,\M_n$ are replaced with progressively simpler structures. 
Our first reduction replaces MSPs with \emph{Choice-Over-Bandits} processes, in which each alternative $i$ consists of picking one bandit process, then sticking with that choice.
We then reduce Choice-Over-Bandits processes to \emph{Choice-Over-Rewards} in which each alternative $i$ consists of picking one random reward to reveal out of several hidden options.
Finally, we reduce Choice-Over-Rewards to random rewards, the setting of prophet inequalities.

\begin{figure}[h]
  \newcommand{\impliesthm}[1]{\parbox[c]{12ex}{\vskip2em \centering $\longrightarrow$ \\ {\centering \scriptsize Prop. \ref{#1}}}}
  \newcommand{\problemeff}[3]{\parbox[t]{#1}{\centering #2 \\ {\scriptsize (#3)}}}
  \newcommand{\problemeffone}[2]{\parbox[t]{#1}{\centering #2 \\ }}

  {\small \hfill
  \addtolength{\tabcolsep}{-1.5ex}
  \begin{tabular}{ccccccc}
    &\problemeff{22ex}{Combinatorial Markov Search}{$\cms$}
    &\impliesthm{prop:cms-to-cobs}
    &\problemeff{14ex}{Choice-Over-Bandits Selection}{$\cobs$}
    &\impliesthm{prop:cobs-to-cors}    
    &\problemeff{14ex}{Choice-Over-Rewards Selection}{$\cors$}
    &\impliesthm{prop:cors-to-prophets}
    \problemeffone{14ex}{Classic Prophets}
  \end{tabular}
  \hfill} 
  \caption{\textbf{Reductions.} Given any downward-closed feasibility constraint and ex-ante $\alpha$ prophet inequality, we obtain an online $\alpha$-approximation algorithm for $\cms$ via a series of reductions.
  } 
  \label{fig:results-overview}
\end{figure}

\subsection{Reduction 1: Markov Search Process to Choice-Over-Bandits}\label{subsec:msp-to-cob}
We begin with the key and most complicated reduction, replacing MSPs with Choice-Over-Bandits processes.
First, we define a \emph{Choice-Over-Bandits} process and the \emph{Choice-Over-Bandits Selection} problem.
\begin{definition}[Choice-Over-Bandits] \label{def:cob}
  A \emph{Choice-Over-Bandits} (\cob{}) process $\C$ consists of $m$ bandit processes $B^1,\dots,B^m$ for some $m \geq 1$.
  An algorithm initially takes a costless action corresponding to some $j \in [m]$.
  From there, the algorithm interacts with the bandit process $B^j$, incurring all transition costs, eventually halting, and possibly claiming a reward.
  We allow the random transitions of $B^1,\dots,B^m$ to be arbitrarily correlated, i.e. to access shared randomness.\footnote{
    Allowing shared randomness will simplify our reduction later.
    We note that a $\cob$ process is not technically a special case of a $\msp$ which, for simplicity, was defined with independent randomness on all transitions.
    However, the definition of $\msp$ could be broadened to allow shared randomness without affecting any of our results, at the cost of a more complex definition.
}  

  The \emph{\cobs{}$(\F)$} problem is defined as the Combinatorial Markov Search$(\F)$ problem where each \msp{} $\M_i$ is replaced by a \cob{} $\C_i$.
  All other definitions, such as welfare and performance, are unchanged.
\end{definition}

The reduction from \cms{} to \cob{} comprises the core of the proof of Theorem \ref{thm:main-1}, and we elaborate here on the \cobs{} problem individually.

The \cobs{} problem restricts the number of choices available in any \msp{} to only one initial choice: which bandit to follow?
This problem alone is an interesting generalization of of the Pandora's Box model, capturing nonobligatory inspection settings (e.g. \citet{doval2018whether}) and matching settings, with appropriate selection of feasibility structures \citep{bowers2024matching}. 
The \cobs{} setting introduces some level of choice over the evolution of each process for the decisionmaker, without allowing arbitrary sets of choices over each \msp. 
This change is sufficient to immediately moves simple single-selection feasibility structures from computationally feasible to NP-hard \citep{fu2023pandora}.

Intuitively, our reduction from a \msp{} to a \cob{} ``detangles" the paths through the \msp, converting any sequence of choices over actions to a single choice of which bandit to pursue (see Figure \ref{fig:alternative-structures}). 
The key to this detangling process is the idea of \emph{induced bandits} (Definition \ref{def:induced-bandit}), which allows us to produce a separate bandit for every deterministic policy on a give \msp, and then choose over those bandits to produce a \cob{} instance. 

\begin{figure}
	\centering
	\begin{subfigure}[b]{0.2\textwidth}
		\centering
		\begin{tikzpicture}[
			wide/.style={line width=4pt}, every node/.style={circle,draw,minimum size=10}, scale=.25]

			\node (root) at (-4,0) {};

			\node (s1) at (-2, 2) {};
			\node (s2) at (-2, -2) {};

			\node  (s3) at (0, 3) {};
			\node  (s4) at (0, 0) {};
			\node  (s5) at (0, -3) {};

			\coordinate  (s6) at (2, 4) {};
			\coordinate  (s7) at (2, 2) {};
			\coordinate  (s8) at (2, 0) {};
			\coordinate  (s9) at (2,-4) {};

			\node (b1) at (6, 4) {};
			\node (b2) at (6, 2) {};
			\node (b3) at (6, 0) {};
			\node (b4) at (6, -2) {};
			\node (b5) at (6, -4) {};

			\draw (root) -- (s1);
			\draw (root) -- (s2);

			\draw (s1) -- (s3);
			\draw  (s1) -- (s4);

			\draw  (s2) -- (s4);
			\draw  (s2) -- (s5);

			\draw (s3) -- (s6);
			\draw (s3) -- (s7);

			\draw (s5) -- (s8);
			\draw (s5) -- (s9);

			\path (s6) -- (b1) node[midway, draw=none,fill=none] {\small $\dots$};
			\path (s8) -- (b3) node[midway, draw=none,fill=none] {\small $\dots$};
			\path (s9) -- (b5) node[midway, draw=none,fill=none] {\small $\dots$};
		\end{tikzpicture}
		\subcaption{MSP}\label{subfig:msp}
	\end{subfigure}
	\hfil
	\begin{subfigure}[b]{0.2\textwidth}
		\centering
		\begin{tikzpicture}[
			wide/.style={line width=4pt}, every node/.style={circle,draw,minimum size=10}, scale=.25]

			\node (root) at (-7, 0) {};

			\node (s1) at (-3.5,4) {};
			\node (s2) at (0,4) {};
			\node (s3) at (4.5,4) {};

			\draw (s1) -- (s2);
			\path (s2) -- (s3) node[midway, draw=none,fill=none] {\small $\dots$};

			\node (s4) at (-3.5,1.3) {};
			\node (s5) at (0,1.3) {};
			\node (s6) at (4.5,1.3) {};

			\draw (s4) -- (s5);
			\path (s5) -- (s6) node[midway, draw=none,fill=none] {\small $\dots$};

			\node (s7) at (-3.5,-1.3) {};
			\node (s8) at (0,-1.3) {};

			\draw (s7) -- (s8);

			\node (s9) at (-3.5,-4) {};
			\node (s10) at (0,-4) {};
			\node (s11) at (4.5,-4) {};

			\draw (s9) -- (s10);
			\path (s10) -- (s11) node[midway, draw=none,fill=none] {\small $\dots$};

			\draw (root) -- (s1);
			\draw (root) -- (s4);
			\draw (root) -- (s7);
			\draw (root) -- (s9);
		\end{tikzpicture}
		\subcaption{\cob{}}\label{subfig:cob}
	\end{subfigure}
	\hfil
	\begin{subfigure}[b]{0.2\textwidth}
		\centering
		\begin{tikzpicture}[
			wide/.style={line width=4pt}, every node/.style={circle,draw,minimum size=10}, scale=.25]

			\node (root) at (-3.5, 0) {};

			\node (s1) at (0,4) {};
			\node (s2) at (0,1.3) {};
			\node (s4) at (0,-1.3) {};
			\node (s7) at (0,-4) {};

			\draw (root) -- (s1);
			\draw (root) -- (s2);
			\draw (root) -- (s4);
			\draw (root) -- (s7);
		\end{tikzpicture}
		\subcaption{\cor{}}\label{subfig:cor}
	\end{subfigure}

	\caption{\textbf{\msp, \cob, and \cor.} A simplified representation of the difference between the action structures used in our reductions, ommitting the randomness of transitions (shown in Figure \ref{fig:msp-random} as dashed lines). 
	Actions advance left to right, with vertices representing states and edges representing actions.
	(\ref{subfig:msp}) A full $\msp$, where costly actions may form any DAG.
	(\ref{subfig:cob}) A Choice-Over-Bandits structure, where after an inital action with cost 0, each state has only one legal action. 
	(\ref{subfig:cor}) A Choice-Over-Rewards structure, where each initial costless action leads to a random reward. 
	}\label{fig:alternative-structures}
\end{figure}

\begin{definition}[Induced bandit] \label{def:induced-bandit}
  Given a \msp{} $\M_i$ and a deterministic policy $\pi$ on $\M_i$, define the \emph{induced bandit process} $B_i^{\pi}$ to be the bandit \msp{} that mimics following $\pi$ on $\M_i$.
  More formally, states in $B_i^{\pi}$ consist of histories $h$ in $\M_i$ reachable by $\pi$; the only legal action at each non-terminal $h$ is ``continue''; and continuing at a non-terminal $h$ has cost and transition probabilities equal to taking the action $\pi(h)$ in $\M_i$.
  If $\pi(h)$ halts or claims, then $h$ is a terminal state in $B$ with reward zero or $V(s)$ respectively, where $s$ is the final state in history $h$.
\end{definition}

\begin{reduction}[\MSPtoCOB{}] \label{reduction:msp-to-cob}
  Given an instance $\I = (\M_1,\dots,\M_n)$ of $\cms(\F)$, the \emph{\MSPtoCOB}$(\I)$ reduction constructs an instance $\I'$ of \cobs$(\F)$ as follows.
  For each MSP $\M_i$ in $\I$, create a Choice-Over-Bandits process $\C_i$ in $\I'$.
  Its choices consist of the bandit processes $B_i^{\pi}$ induced by every deterministic policy $\pi$ on $\M_i$ (of which there are finitely many).
\end{reduction}
We note that, for this reduction, it is necessary in the definition of \cob{} that the bandits may have correlated transitions.

\begin{restatable}{lemma}{cobstocms} \label{lemma:alg-cobs-cms}
  Let $\I$ be an instance of \cms{} and $\I' = $\MSPtoCOB$(\I)$.
  For any policy $\rho'$ on $\C_i$, there exists a policy $\rho$ on $\M_i$ such that $\Perf^{\rho}(\M_i) = \Perf^{\rho'}(\C_i')$ and the probabilities that $\rho$ and $\rho'$ claim are the same.
  Furthermore, given any policy $\rho$ on $\M_i$, there exists a policy $\rho'$ on $\C_i$ with the same guarantee.
  
\end{restatable}
\begin{proof}[Sketch (proof in Section \ref{app:inefficient})]
  $\rho'$ first selects one of the bandits $B_i^{\pi}$ in $\C_i$.
  At each subsequent step, $\rho'$ either advances $B_i^{\pi}$, claims $\C_i$, or halts. 
  We define $\rho$ to always make the same choice in $\M_i$, i.e. $\pi(h)$ at the current history $h$, claim, or halt.
  By construction, the choice is always legal in $\M_i$ at the current history $h$.
  The costs incurred and reward claimed, if any, are the same.

  To prove the converse, we first observe that a randomized policy $\rho$ for $\M_i$, WLOG, first draws a deterministic policy $\pi$ from some probability distribution and then follows it.
  Define $\rho'$ to draw $\pi$ from this distribution and choose the bandit $B_i^{\pi}$, always advancing unless at a terminal state with reward zero, in which case $\rho'$ halts.
  Then $\rho'$ simulates the behavior of $\rho$ on $\M_i$.
\end{proof}

\begin{restatable}[\cms{} to \cobs{}]{proposition}{propcmstocobs} \label{prop:cms-to-cobs}
Let $\F$ be downward-closed.
If there exists an online ex-ante $\alpha$-approximation algorithm $\ALG'$ for \cobs$(\F)$, then there exists an online ex-ante $\alpha$-approximation algorithm $\ALG$, not necessarily efficient, for $\cms(\F)$.  
\end{restatable}
\begin{proof}[Sketch (proof in Section \ref{app:inefficient})]
  Lemma \ref{lemma:alg-cobs-cms} implies the following claim: any algorithm for $\cobs$ that consists of executing a policy on $\C_1$, then a policy on $\C_2$, and so on, can be converted to an algorithm for $\cms$ with the same approximation factor. The reverse is also true.
  In particular, this claim applies to the online algorithm $\ALG'$.
  Also, by Observation \ref{obs:ex-ante-indep}, the optimal ex-ante feasible algorithms $\OPT'$ for $\cobs$ and $\OPT$ for $\cms$ consist WLOG of collections of independent policies for each alternative, so the above claim applies to $\OPT$ and $\OPT'$ as well.
\end{proof}

Note that the number of deterministic policies on a MSP $\M$ is at least exponential in the size of the state graph of $\M$, so the above reduction is not efficient.
In Section \ref{sec:efficient}, we will use the same reduction for analysis, but give an algorithm that efficiently selects a policy without constructing the reduction explicitly.

\subsection{Reduction 2: Choice-Over-Bandits to Choice-Over-Rewards}\label{subsec:cob-to-cor}

We next reduce to the simpler case where each alternative $i$ is a \emph{Choice-Over-Rewards}.

\robin{begin edits}
\begin{definition}[Choice-Over-Rewards]  \label{def:cor}
  A \emph{Choice-Over-Rewards} (\cor{}) process $\C_i$ consists of $m$ random variables $X_i^1,\dots,X_i^m$ for some $m \geq 1$.
  An algorithm initially chooses $j \in [m]$ and observes $X_i^j$.
  The algorithm then may decide either to discard $i$; or to claim $i$, obtaining reward $X_i^j$.
  We allow $X_i^1,\dots,X_i^m$ to be arbitrarily correlated with each other (but independent of any $X_{i'}^j$ for $i' \neq i$).

  The \emph{\cors{}$(\F)$} problem is defined as the Combinatorial Markov Search$(\F)$ problem where each \msp{} $\M_i$ is replaced by a \cor{} $\C_i$.
  The performance on $\C_i$ is simply the reward $X_i^j$ claimed at $i$ if any; and the welfare is the sum of claimed rewards.
\end{definition}

Here, we may picture each $i$ as a ``cabinet'' consisting of $m$ ``drawers'', inside each of which is a random prize.\footnote{Thanks to Bobby Kleinberg for giving this metaphor for the problem.}
An algorithm may only open one of the drawers on each cabinet.\footnote{Thus, the \cors{} problem is a type of stochastic probing problem. We have not found a reference to it specifically in the literature.}

Our next reduction is from Choice-Over-Bandits to Choice-Over-Rewards.
Our reduction amortizes the transition costs of the bandit processes using techniques from the Pandora's Box literature ~\citep{kleinberg2016descending,bowers2024matching}.
Essentially, we condense each bandit process into a representative index representing an expected overall reward for continuing to advance the bandit. 
This indexing allows us to make an initial decision over ``rewards" summarizing the expected benefit of each bandit using a Choice-Over-Rewards algorithm, and combining this with known bandit approaches we obtain an algorithm for Choice-Over-Bandits. 

\robin{end edits}

\begin{reduction}[\COBtoCOR{}] \label{reduction:cob-to-cor}
 Given a Choice-Over-Bandits $\C_i$ consisting of bandits $B_i^1,\dots,B_i^m$, we let the Choice-Over-Rewards instance $\C_i'$ consist of random rewards $\kappa_i^1,\dots,\kappa_i^m$ where $\kappa_i^j$ is the capped value of $B_i^j$ (Definition \ref{def:weitzman}).
  Given an instance $\I = (\C_1,\dots,\C_n)$ of \cobs$(\F)$, the \emph{\COBtoCOR}$(\I)$ reduction constructs an instance $\I'$ of \cors$(\F)$ consisting of $(\C_1',\dots,\C_n')$.
\end{reduction}

We use the fact that online algorithms for \cors{} are, WLOG, ``quantile-based''.
This is related to the the ``monotone'' property of prophet-inequality algorithms~\citet{kleinberg2012matroid}.

\begin{definition}[Quantile-based] \label{def:quantile}
  An algorithm for \cors{} is \emph{quantile-based} if, for each alternative $i$ and drawer $j$, there exists $q_i^j \in [0,1]$ such that, conditioned on the algorithm opening drawer $j$ on alternative $i$, the algorithm claims $i$ if and only if $X_i^j$ is in its top $q_i^j$ quantile of realizations.\footnote{For a discrete random variable, this may involve a randomized decision; for example, given a Bernoulli$(\tfrac{1}{2})$ reward and $q_i^j = 0.25$, the reward should be taken half of the time when its value is $1$ and never when its value is $0$.}
\end{definition}

\begin{observation} \label{obs:quantile-based}
  For any online algorithm for \cors{}, there exists a quantile-based online algorithm with at least as large expected performance and the same feasibility guarantee.
\end{observation}
\begin{proof}
For any online algorithm which claims $X_i^j$ with probability $q_i^j$, we can maintain feasibility while moving all probability of claiming to the top $q_i^j$ quantile of reward realizations.
Because the algorithm is online, the criteria for selection does not affect our selection of any preceding alternatives.
And as the probability of selecting $i$ remains the same, following alternatives can be claimed with the same probability. 
\end{proof}

We now complete this portion of the reduction exactly as the previous portion: showing that policies for $\cor$ can be converted to policies for $\cob$, then ``lifting'' this guarantee to algorithms that consist of running policies on the alternatives one at a time.

\begin{lemma} \label{lemma:alg-cors-cobs}
  Let $\I$ be an instance of \cobs{} and $\I' = $\COBtoCOR$(\I)$.
  For any policy $\rho'$ on $\C_i'$, there exists a policy $\rho$ on $\C_i$ such that $\Perf^{\rho}(\C_i) \geq \Perf^{\rho'}(\C_i')$ and the probabilities that $\rho$ and $\rho'$ claim are the same.
  Given any policy $\rho$ on $\C_i$, the reverse is also true.
\end{lemma}
\begin{proof}
  Given policy $\rho'$ on the Choice-Over-Rewards $\C_i'$, by Observation \ref{obs:quantile-based}, WLOG it is quantile-based with parameters $\{q_i^j\}$.
  Define $\rho$ on the Choice-Over-Bandits $\C_i$ to initially make the same choice $j$ as $\rho'$.
  Here, treat $j$ as a random variable, the choice.
  Let $\rho$ be the non-exposed policy on bandit $B_i^j$ that claims if $\kappa_i^j$ is in its top $q_i^j$ quantile of realizations; it is constructed in Lemma \ref{lemma:non-exposed-policy}.
  By non-exposure and Lemma \ref{lem:magic}, $\E[\Perf^{\rho}(\C_i)] = \E[\A_i^{\rho} \kappa_i^j] = \E[\A_i^{\rho'} X_i^j] = \E[\Perf^{\rho'}(C_i')]$.

  For the reverse direction, given policy $\rho$ on the \cob{} $\C_i$, let $q_i^j$ be the probability with which it claims conditioned on making choice $j$.
  Define $\rho'$ to make the same choice $j$ as $\rho$, then claim if $X_i^j$ is in its upper $q_i^j$ quantile.
  Because $X_i^j = \kappa_i^j$ in distribution and by definition of $q_i^j$, we have $\E[ \A_i^{\rho'} X_i^j] \geq \E[ \A_i^{\rho} \kappa_i^j]$.
  Then, using Lemma \ref{lem:magic}, $\E[\Perf^{\rho'}(\C_i')] = \E[\A_i^{\rho'} X_i^j] \geq \E[ \A_i^{\rho} \kappa_i^j] \geq \E[ \Perf^{\rho}(\C_i)]$.
\end{proof}

\begin{proposition}[\cobs{} to \cors{}] \label{prop:cobs-to-cors}
  Let $\F$ be downward-closed.
  If there exists an online, ex-ante $\alpha$-approximation for \cors$(\F)$, there exists an online, ex-ante $\alpha$-approximation for \cobs$(\F)$.
\end{proposition}
\begin{proof}
  Exactly analogous to Proposition \ref{prop:cms-to-cobs}, substituting Lemma \ref{lemma:alg-cors-cobs} for Lemma \ref{lemma:alg-cobs-cms}.
\end{proof}

\subsection{Reduction 3: Choice-Over-Rewards to Prophets} \label{subsec:prophets}
Finally, we reduce from \cors$(\F)$ to prophet inequalities, which address the online version of the following problem.

\begin{definition}[\vars] \label{def:selection}
An instance of the \vars$(\F)$ problem consists of a set $X = (X_1,\ldots,X_n)$ of independent random variables.
The algorithm's goal is to select the set $F \in \F$ maximizing $\sum_{i \in F} X_i$.
\end{definition}
An online ex-ante $\alpha$-approximation for \vars$(\F)$ is called an  \emph{ex-ante $\alpha$ prophet inequality}.

\robin{begin edits}
This result relies on observations of the structure of the optimal ex-ante algorithm for the \cors{} problem. 
We observe that the optimal ex-ante \cors{} algorithm induces a distribution over rewards chosen for each Choice-Over-Rewards structure, independently and uncorrelated to the realization of each reward. 
Furthermore, the ex-ante algorithm's probability of selecting each reward can be computed and combined with the realizations of each reward to simulate a Prophets instance with a single reward. 

\robin{end edits}
\begin{proposition}[\cor{} to \vars{}] \label{prop:cors-to-prophets}
Let $\F$ be a non-empty, downward-closed set system.
If there exists an ex-ante $\alpha$ prophet inequality for $\F$, there exists an online ex-ante $\alpha$-approximation for $\cors(\F)$.
\end{proposition}
\begin{proof}
The optimal ex-ante feasible algorithm $\OPT$ for \cors$(\F)$ induces some $Q \in\PF$ where $Q_i$ is the probability of claiming $\M_i$.
Given $Q_i$, WLOG by Observation \ref{obs:ex-ante-indep}, the selection decision on $\C_i$ is statistically independent of all actions on other alternatives $\C_{i'}$.
Therefore, also WLOG, $\OPT$ decides which drawer $j^*(i)$ to open on $\C_i$ independently according to some distribution $\lambda_i\in\Delta_{[m_i]}$.

Let $Y_i$ be a random variable distributed as $Y_i = X_i^{j^*(i)}$, or in other words, $Y_i = X_i^j$ with probability $\lambda_i(j)$.
The $Y_i$s are independent, and $\Welf^{\OPT}(X)=\Welf^{\OPT'}(Y)$, where $\OPT'$ is the optimal ex-ante feasible algorithm for \vars$(\F)$.

Given a prophet inequality, i.e. given an online ex-ante $\alpha$-approximation $\ALG'$ for \vars$(\F)$, and given $\Lambda=(\lambda_1,\ldots,\lambda_n)$, we construct an online algorithm $\ALG$ for \cors$(\F)$.
We instantiate $\ALG'$ on the distributions of $Y_1,\dots,Y_n$.
On arrival $i$, we open drawer $j$ with probability $\lambda_i(j)$.
We set $Y_i = X_i^j$ and claim alternative $i$ if and only if $\ALG'$ claims on $Y_i$.
Our algorithm $\ALG$ is ex-post feasible because $\ALG'$ is, and its performance on each arrival equals that of $\ALG'$.
So
\[
\E[\Welf^{\ALG}(X)] = \E[\Welf^{\ALG'}(Y)] \geq \alpha\cdot \E[\Welf^{\OPT'}(Y)] = \alpha\cdot \E[\Welf^{\OPT}(X)] .
\]
\end{proof}

\begin{proof}[Proof of Theorem \ref{thm:main-1}]
  Let $\F$ be downward-closed.
  By Proposition \ref{prop:cors-to-prophets}, if there is an ex-ante $\alpha$ prophet inequality for $\F$, then there is an online ex-ante $\alpha$-approximation algorithm for \cors$(\F)$.
  By Proposition \ref{prop:cobs-to-cors}, there is therefore one for \cobs$(\F)$, and by Proposition \ref{prop:cms-to-cobs}, therefore one for \cms$(\F)$.
\end{proof}
\citet{lee2018optimal} provide a $\frac{1}{2}$-approximation ex-ante prophet inequality for matroids: it follows that there exists a $\frac{1}{2}$ approximate online algorithm for \cms$(\F)$.
They additionally provide a $\left(1-\tfrac{1}{e}\right)$-approximation ex-ante prophet inequality for matroids when arrival order is stochastic, which translates to the same guarantee for \cms$(\F)$ as well.

\section{Efficient Approximation Algorithm}  \label{sec:efficient}

In this section, we prove Theorem \ref{thm:main-2}: existence of an \textbf{efficient} $(\tfrac{1}{2}-\epsilon)$-approximation algorithm for $\cms(\F)$, for every $\epsilon > 0$, whenever $\F$ forms a matroid.
To do so, we will define algorithms of a particular form.
They turn out to be \emph{incentive-compatible} in the mechanism design sense (see Subsection \ref{subsec:mech-design}), and incentive-compatibility will be crucial to achieving an efficient approximation algorithm.

\subsection{Threshold-Based Algorithms} \label{subsec:threshold}

Traditional prophet inequality algorithms often set a fixed threshold for each arriving random reward, accepting the reward if and only if it exceeds the threshold.
In mechanism-design settings, we can view the threshold as a posted price.
The arrival is an agent with a random value for ``winning'' in the auction, i.e., being included in the final solution $F \subseteq [n]$.
The agent chooses to ``buy'', i.e. be included in $F$, if and only if her value exceeds the posted price.

Our reductions in Section \ref{sec:inefficient} converts such prophet inequalities into algorithms that are not incentive-compatible\footnote{The reduction in Section \ref{sec:inefficient} takes a threshold-based prophet inequality, computes the probability $q_i$ with which it claims an arrival $i$, then selects a policy $\pi$ for $\M_i$ randomly from a certain distribution chosen to maximize ``utility'' subject to claiming with probability $q_i$.} nor computationally efficient.
We will address both problems by defining \emph{threshold-based} algorithms in the Combinatorial Markov Search setting.
An alternative $\M_i$ is viewed as an agent who interacts with a Markov Search Process to discover their value for ``winning'' in the auction.
Given a threshold $\tau$, the agent optimally explores $\M_i$ and decides whether or not to ``buy'' the revealed reward for price $\tau$.
This exploration may involve abandoning the search at any point if the expected utility for continuing becomes negative given that the final reward will cost $\tau$ to claim.
We emphasize that, for now, the agent is merely a thought experiment used to define our algorithm.

\begin{definition}[SAUP, threshold-based]  \label{def:saup}
  Given a single Markov Search Process $\M$ and a threshold $\tau$, the \emph{Single-Agent Utility Problem (SAUP)} is to compute a policy $\pi$ on a given MSP $\M$ that maximizes $\E[\Perf^{\pi}(\M)-\A^{\pi}_{}\tau]$.
  An algorithm for Combinatorial Markov Search is \emph{threshold-based} if, for each arrival $\M_i$, it computes a threshold $\tau_i$ based on the problem instance and on previous decisions, then solves SAUP$(\M_i,\tau_i)$.
\end{definition}
For example, suppose $\M_i$ consists of a Pandora's box, i.e. there is a single action with cost $c_i$ that leads to a stochastic reward.
In this case, SAUP$(\M_i,\tau_i)$ is the following policy: if $\tau_i$ is greater than the Weitzman index $\sigma_i$ of the box, then discard the box immediately without paying to open it; otherwise, open the box, discard the reward if it is smaller than $\tau_i$, and otherwise claim the reward.
We now show that SAUP can be computed efficiently for a generic Markov Search Process.
A similar result was derived independently by \citet{chawla2025combinatorial}.

\begin{proposition} \label{prop:saup}
  There is a polynomial-time algorithm for SAUP$(\M,\tau)$.
\end{proposition}
The algorithm computes a policy $\pi^*$ by backward induction as follows.
\begin{itemize}
  \item At a terminal state $s$: $\pi^*$ claims if $V(s) \geq \tau$, else halts.
        Let $\Val^{\pi^*,s} := (V(s) - \tau)^+$.
  \item At a non-terminal state $s$: Let $p_{a,s} := P(~\cdot~;a,s)$ be the probability distribution over the next state when taking action $a$.
        Define $v^{\pi^*}(a;s) := \E_{s' \sim p_{a,s}} [ \Val^{\pi^*,s'} ] - C(a,s)$.
        Let $a^* = \arg\max_a v^{\pi^*}(a;s)$ and set $\pi^*(s) = a^*$ unless $v^{\pi^*}(a^*;s) \leq 0$, in which case $\pi^*(s) = $ halt.
        Let $\Val^{\pi^*,s} = 0$ if $\pi^*(s) = $ halt, else $\Val^{\pi^*,s} = v^{\pi^*}(a^*,s)$.
\end{itemize}
Intuitively, at state $s$, the algorithm inductively calculates a policy for each possible successor of $s$, as well as its expected reward,
and selects the action $a^*$ with the highest expected reward, then takes action $a^*$ or halts if its expected reward is nonpositive. 

\begin{proof}[Proof of Proposition \ref{prop:saup}]
  For any deterministic policy $\pi$, consider running $\pi$ on $\M$ starting from state $s$.
  Let $\A^{\pi,s}$ be the indicator that $\pi$ claims and let $\Perf^{\pi,s}$ be the performance.
  The objective for $\SAUP(\M,\tau)$ is $\E[ \Perf^{\pi,s} - \A^{\pi,s}\tau ]$.
  Observe that for $\pi=\pi^*$, this is equal to $\Val^{\pi^*,s}$. \elias{say more?}
  We show $\Val^{\pi,s} \leq \Val^{\pi^*,s}$ for all $s$ by backward induction.
  At terminal states $s$, $\Val^{\pi,s} \leq (V(s) - \tau)^+ = \Val^{\pi^*,s}$.
  At non-terminal states $s$, if $\Val^{\pi,s} \leq 0$ then $\Val^{\pi,s} \leq \Val^{\pi^*,s}$ immediately.
  Otherwise:
  \begin{align*}
    \Val^{\pi,s}
    &=    \E_{s' \sim p_{\pi(s),s}} [ \Val^{\pi,s'} ] - C(\pi(s),s)  \\
    &\leq \E_{s' \sim p_{\pi(s),s}} [ \Val^{\pi^*,s'} ] - C(\pi(s),s)  & \text{inductive hypothesis}  \\
    &\leq \E_{s' \sim p_{\pi^*(s),s}} [ \Val^{\pi^*,s'} ] - C(\pi^*(s),s)  & \text{construction of $\pi^*$}  \\
    &= \Val^{\pi^*,s} .
  \end{align*}
  In particular, when $s$ is the start state, this proves that $\pi^*$ is optimal among deterministic policies.
  Any randomized algorithm can be written as a distribution over deterministic policies, so $\pi^*$ is optimal among all policies.

  For each state $s$ of $\M$ and for each action $a\in A^s$, the algorithm performs a polynomial-time computation to determine $v^{\pi^*}(a,s)$.
  Therefore, the algorithm runs in polynomial time in the size of the state graph of $\M$.
\end{proof}

\subsection{A Matroid ``Prophet'' Algorithm} \label{subsec:matroid-prophet}

We design an online threshold-based algorithm.
It has two components: computing the threshold $\tau_i$ for each arrival $\M_i$, and running SAUP$(\M_i,\tau_i)$ to interact with $\M_i$ and determine whether to claim it.
We have already given a polynomial-time algorithm for SAUP$(\M_i,\tau_i)$ in Proposition \ref{prop:saup}, so we now turn to computing the thresholds.

\paragraph{Computing thresholds and approximating ex-ante OPT.}
Matroid prophet algorithms such as in \citet{kleinberg2012matroid,lee2018optimal} require knowledge of the optimal fully-adaptive solution in order to compute thresholds.
While we do not know how to compute or even approximate OPT in our setting, we will approximate the ex-ante feasible OPT, which is an upper bound on the ex-post feasible OPT.

Approximation, however, presents an additional challenge because existing algorithms of \citet{kleinberg2012matroid,lee2018optimal} cannot obviously be adapted to work when OPT is only approximated.\footnote{The thresholds in these algorithms are based on differences in OPT across different inputs, so an approximation guarantee for OPT does not immediately translate to any approximation guarantee for those algorithms.}
Our algorithm addresses this in part by requiring only a single call to an approximation algorithm for the ex-ante feasible $\OPT$.
We then base all of the algorithm's thresholds on this solution.
More precisely, we require the approximation algorithm to output the following.

\begin{restatable}[$\epsilon$-approximate collection]{definition}{defepsapproxcoll} \label{def:eps-approx-coll}
  Given an input $\M_1,\dots,\M_n$ to the ex-ante $\cms(\F)$ problem, an \emph{$\epsilon$-approximate collection} is a pair of vectors $(Q,\hat{U})$ satisfying the following:
\begin{enumerate}
  \item $Q \in \PF$, the feasible polytope.
  \item there exists an ex-ante feasible algorithm $\OPTHAT$ that is a $1-\epsilon$ approximation and satisfies $Q_i \geq \Pr[i \in \OPTHAT]$ for all $i$.
  \item letting $U_i$ be the expected utility $\OPTHAT$ obtains on $\M_i$, $\hat{U}$ approximates $U$ from below, i.e. $\hat{U}_i \leq U_i$ for all $i$ and $\sum_i \hat{U}_i \geq (1-\epsilon)\sum_i U_i$.
\end{enumerate}
\end{restatable}
We can think of $(Q,\hat{U})$ as a witness for the existence of an approximately-optimal ex-ante feasible algorithm.
We allow $\hat{U}$ to be approximate because it may be computationally intractable to compute $U$ exactly.

\begin{restatable}{proposition}{propexanteapprox} \label{prop:ex-ante-approx}
  When $\F$ is a matroid, there is an algorithm for ex-ante $\cms(\F)$ that ouputs an $\epsilon$-approximate collection in time polynomial in the input size and $\frac{1}{\epsilon}$.

\end{restatable}
The full proof is contained in Subsection \ref{subsec:ex-ante-approx-pf}. The algorithm uses backward induction on each MSP to approximate, at each state, the concave nondecreasing function $f(q)$ mapping a probability $q$ to the maximum expected utility that can be generated, starting from that state, subject to claiming with probability at most $q$.
Convex programming is used at each step, with the accumulation in error carefully bounded.
We initially give an additive-multiplicative approximation guarantee, then round small enough values down to zero, obtaining the $\epsilon$-approximate collection.
Using the same technique, we are also able to give a Fully-Polynomial Time Approximation Scheme (FPTAS) for ex-ante $\cms(\F)$ whenever $\F$ is a matroid, presented in Proposition \ref{prop:fptas-itself}.

\paragraph{The matroid prophet algorithm.}
Our online algorithm is based on the matroid prophet approach of \citet{kleinberg2012matroid,lee2018optimal}, in which thresholds $\tau_i$ are based on the expected marginal contribution of arrival $\M_i$ to the optimal solution that is constrained to coincide with the algorithm's choices so far (i.e. on arrivals $1,\dots,i-1$).
The outline of the algorithm is as follows.
\begin{enumerate}
  \item Compute an $\epsilon$-approximate collection $(Q,\hat{U})$ to the ex-ante optimum.
  \item Decompose $Q$ into a probability distribution $D_Q$ over feasible subsets $F \in \F$.
        This turns out to be achievable in polynomial time because matroid polytopes have efficient separation oracles.
  \item On each arrival $\M_i$, use $D_G$ and $\hat{U}$ to compute a threshold $\tau_i$.
        The formula for $\tau_i$ represents an ``expected contribution'' of $\hat{U_i}$ to the subset of $\{1,\dots,i-1\}$ claimed by our algorithm so far, intersected with a random feasible subset drawn from $D_Q$.
        Run \SAUP$(\M_i,\tau_i)$ to interact with $\M_i$ and determine whether to claim it.
\end{enumerate}

See Algorithm \ref{alg:matroid-prophet} for full details.

\begin{restatable}{proposition}{propmatroidprophet} \label{prop:matroid-prophet}
  Let $\F$ be a matroid.
  Given $\M_1,\dots,\M_n$, let $(Q,\hat{U})$ be an $\epsilon$-approximate collection for any $\epsilon \geq 0$.
  There is a polynomial-time online algorithm $\ALG$ for $\cms(\F)$ that takes as input $(Q,\hat{U})$, interacts with each arriving $\M_i$ only through calls to \SAUP$(\M_i,\tau)$, and guarantees $\E[\ALG] \geq \tfrac{1}{2}\sum_{i=1}^n \hat{U}_i$.
\end{restatable}

The full proof is contained in Subsection \ref{subsec:matroid-prophet-pf}. We note that the algorithm takes any feasible $(Q,\hat{U})$ and provides a guarantee relative to $\hat{U}$.
The better the algorithm $\OPTHAT$ associated with $Q$, and the better $\hat{U}$ approximates the true performances $U$ of that algorithm, the better the approximation.
In the case $\epsilon = 0$, we obtain a $\tfrac{1}{2}$ approximation.

To prove Proposition \ref{prop:matroid-prophet}, we reduce to the Choice-Over-Reward Selection setting.
There, we consider a standard prophet inequality decomposition of performance into ``revenue'' (interpreted as payments of agents for ``winning'') and ``utility'' of the agents.
Although Choice-Over-Reward Selection appears more complex than traditional prophets settings, it remains true that, by solving \SAUP{}, the ``utility'' of our algorithm on a given arrival exceeds the ``utility'' of $\OPT$ against the same ``posted price'' threshold.
We can then continue with a somewhat standard analysis, although care is needed to adapt known techniques.
For example, the algorithm utilizes an expectation over ``Bernoullified'' rewards for efficiency reasons, but the proof needs to consider ``re-randomized'' rewards that are converted to Bernoullified rewards at just the right point in the analysis.

We now have all the components to prove our second main result, an efficient $(\tfrac{1}{2}-\epsilon)$-approximation for $\cms(\F)$ for any matroid $\F$.

\begin{proof}[Proof of Theorem \ref{thm:main-2}]
  Let $\OPT$ be the optimal ex-ante feasible solution, which upper-bounds the optimal ex-post feasible solution.
  Given $\epsilon > 0$ and $\M_1,\dots,\M_n$, using Proposition \ref{prop:ex-ante-approx}, we compute in time $\text{poly}(\M_1,\dots,\M_n,\tfrac{1}{\epsilon})$ an $\epsilon$-approximate collection $(Q,\hat{U})$.
  Then, we execute the online algorithm of Proposition \ref{prop:matroid-prophet} with $\M_1,\dots,\M_n$ and $(Q,\hat{U})$ as input.
  The online algorithm runs in polynomial time, making $n$ calls to our \SAUP{} algorithm of Proposition \ref{prop:saup}, each of which runs in polynomial time.

  The algorithm's performance satisfies
  \begin{align*}
    \E[\ALG] &\geq \frac{1}{2}\sum_{i=1}^n \hat{U}_i  & \text{Proposition \ref{prop:matroid-prophet}}  \\
             &\geq \frac{1-\epsilon}{2} \sum_{i=1}^n U_i   & \text{Definition \ref{def:eps-approx-coll}} \\
             &\geq \frac{(1-\epsilon)^2}{2}  \OPT  & \text{Definition \ref{def:eps-approx-coll}}  \\
             &\geq \left(\tfrac{1}{2} - \epsilon\right) \OPT .
  \end{align*}
\end{proof}

\subsection{Application to Mechanism Design} \label{subsec:mech-design}

Consider a setting with $n$ strategic agents and a resource to be allocated.
There is a constraint $\F \subseteq 2^{[n]}$ specifying which subsets $S$ of agents may be feasibly allocated the resource at the same time.
Each agent $i$ interacts with an independent, private Markov Search Process $\M_i$ representing investment, research, and development regarding how to utilize the resource.
The costs of the MSP are incurred by the agent, and the available reward of the MSP is the value generated for the agent if they are allocated, i.e. included in $S$.
An agent's utility is their reward if allocated (else zero), minus the sum of all costs incurred in their MSP, minus any payment the agent makes.
We call the version of the $\cms$ problem where agents control each MSP strategic the \cms{}-Auction problem. 

A \emph{mechanism} is a game played by the $n$ agents where an outcome consists of an allocation $S \in \F$ along with payments charged to each agent.
We assume the agents must reach a terminal state of the MSP and reveal a reward to be allocated.
A mechanism has a \emph{Price of Anarchy (PoA)} of $\alpha$ if, in every Nash equilibrium, the expected sum of agent utilities and the payments is at least $\alpha$ times the expected performance of the optimal algorithm for $\cms(\F)$ on $(\M_1,\dots,\M_n)$.

We now prove Corollary \ref{cor:main-poa}, yielding a Price of Anarchy of $1/2$ for this problem (or $1/2-\epsilon$ for a mechanism that runs in polynomial time).
We note that in all cases, the agents' equilibrium strategies can be computed in polynomial time.

\begin{proof}[Proof of Corollary \ref{cor:main-poa}]
The mechanism is a sequential posted-price mechanism.
We approach the agents in order $1,\dots,n$ one at a time and offer each $i$ a price $\tau_i$ for being included in the allocation.
We compute $\tau_i$ using Theorem \ref{thm:main-2}.
It is a unique best-response for $i$ to solve SAUP$(\M_i,\tau_i)$, i.e. to implement the current stage of our algorithm.
The performance of the algorithm is equal to the sum of agent utilities and payments, which implies the result.
We note that, to inefficiently obtain a factor exactly $\tfrac{1}{2}$, we can compute the ex-ante optimal policy exactly and use it with the matroid prophet algorithm constructed in Proposition \ref{prop:matroid-prophet} instead of using the approximation to OPT.
\end{proof}

\section{Related work} \label{sec:related}

This paper's technical contributions bridge the broad literature on prophet inequalities and mechanism design with the smaller but still substantial literature on Pandora's Box.
Conceptually, our contributions are also related to work on bandit superprocesses.

\subsection{Superprocesses and index theorems}

A superprocess or \emph{alternative decision process} consists of a set of Markov Decision Processes (MDPs), which generate costs or rewards with each action taken \citep{nash1973optimal}. 
This differs from our MSPs in the presence of \emph{rewards} as well as costs for advancement, rather than having a final reward for selection of a process.  
Nevertheless, our work is closely related conceptually to work on superprocesses, which generally focuses on sufficient conditions for index theorems (efficient local rules that enable globally optimal algorithms).
The most famous is the Gittins index~\citep{gittins1979bandit}, but many others along this line have been proposed and studied, some allowing for limited decisionmaking or other relaxations of the bandit setting~\citep{glazebrook1976profitability,whittle1980multi,nash1980generalized,glazebrook1982sufficient,granot1991optimal,glazebrook1993indices,dumitriu2003playing,keller2003branching}.
\citet{scully2025gittins}'s recent tutorial on the Gittins index also highlights the close connection to the Pandora's Box. 

Applications in AI and reinforcement learning include \citep{brown2013optimal,hadfield-menell2015multi}.
Our main question is relevant to this literature: we ask, in a somewhat different setting, how to construct robust and approximate local strategies when index theorems are unavailable.

\citet{guha2013approximation} considers generalizations of the Gittins setting in which an index theorem does not apply, obtaining constant-factor approximations using techniques such as linear programming relaxations.
\citet{ma2018improvements} obtains constant-factor approximation algorithms for broad structures of superprocesses in a finite-horizon setting.
These results can be related to versions of our problem with single selection and where actions always yield rewards, rather than costs as in our model.

\subsection{Prophet inequalities}

There is a substantial existing literature on prophet inequalities.
\citet{krengel1978semiamarts} first proposed prophet inequalities, and \citet{samuelcahn1984comparison} first proposed a thresholding algorithm.
Various extensions allowing for the selection of multiple items, particularly under matroid constraints, have been studied by  \citet{chawla2010multi}, \citet{yan2011mechanism}, \citet{kleinberg2012matroid}, \citet{feldman2016online}, \citet{dutting2020prophet}, and \citet{chawla2024non}.
\citet{lee2018optimal} used ex-ante relaxations to analyze prophet inequalities with matroid selection constraints, a technique we employ heavily.
Other extensions of prophet inequalities have also been analyzed.
\citet{rubinstein2017combinatorial} studies combinatorial valuations of the items, \citet{ezra2023prophet} studies a different objective (expected ratio rather than ratio of expectations), and \citet{immorlica2020prophet} studies correlated values.

Prophet inequalities have natural applications to incentive-compatible mechanism design, in particular posted price auctions (\citet{chawla2010multi} and subsequent literature; see \citet{lucier2017economic}).
\citet{alaei2022descending} and \citet{kleinberg2016descending} study connections to descending price auctions.
Many of the aforementioned papers discuss these economic connections, and \citet{lucier2017economic} provides a survey.

\subsection{Pandora's box}
The Pandora's Box problem -- and its polynomial-time, optimal solution -- were posed by \citet{weitzman1979optimal}.
For a full survey of Pandora's Box problems and recent developments, we refer to \citet{beyhaghi2024recent}.
Various models of multistage Pandora's Box, and similar settings, have been studied in \citet{guha2007information, aouad2020pandora,gupta2019markovian,ke2019optimal, bowers2024matching}.
Another well-studied variation on Pandora's Box which figures in our work is ``non-obligatory inspection,'' first posed by \citet{doval2018whether}. 
\citet{beyhaghi2019pandora} provides a $1-\frac{1}{e}$-approximation for the problem using techniques from \citet{asadpour2016maximizing}, and \citet{scully2025localhedgingapproximatelysolves} propose an approximation algorithm via a randomized ``committing" algorithm.
\citet{fu2023pandora} established that non-obligatory Pandora's Box is an NP-hard problem, unlike the original and multistage versions of the problem and, along with \citet{beyhaghi2023pandora} independently, provide a PTAS.
In Appendix \ref{app:hardness}, we extend the results of \citet{fu2023pandora} to prove hardness of \cms.  
Other extensions of Pandora's Box have been studied such as combinatorial selection constraints in \citet{singla2018price}, order constraints in \citet{boodaghians2020pandora}, correlations between boxes in \citet{chawla2020pandoras}, and applications to matching and mechanism design in \citet{immorlica2020information},  \citet{bowers2023high}, and \citet{bowers2024matching}.

Several recent works are especially close to our model and techniques. 

\citet{gupta2019markovian} studies what we call the \emph{bandits} special case of our problem in an offline setting, showing, among other results, that ``frugal'' algorithms for combinatorial optimization extend with the same approximation guarantees.
We focus on a significantly more general \msp{} model that allows at each time for a choice between multiple possible decisions with different costs and transition functions.
We do not know of works in this literature that have considered more adaptive decisionmaking processes such as our cabinets and DAGs models.

\citet{esfandiari2019online} considers a problem where Pandora boxes arrive online and apply prophet inequality techniques.
Their closest result to this paper involves a setting inspired by reinforcement learning in which there are a number of rounds, and in each round, a number of Pandora boxes arrive; only one can be opened per round.
This setting is a case of the $\cobs$ special case of our problem with each bandit consisting only of a Pandora box.
\citet{esfandiari2019online} makes several additional assumptions: the set of arriving boxes have the same parameters (i.e. cost and value distributions) in every round, and the benchmark ``optimal'' algorithm is restricted and not fully adaptive as in our paper.
Their constraint, that at most one box of each ``type'' can be claimed, is generally incomparable to our matroid constraint results, but in other respects their result can be viewed as a highly special case of Theorem \ref{thm:main-2}.
Our key technical tools, such as our ``matroid prophet inequality'' for the $\cors$ setting, our reductions from $\cms$ to $\cobs$, our use of the $\SAUP$ subroutine and threshold-based algorithms, and our extensive use of the \emph{ex-ante} benchmark, are not related to any of the techniques in \citet{esfandiari2019online}.

Independently to this work, \citet{chawla2025combinatorial} introduce the \emph{Costly Information Combinatorial Selection (CICS)} problem, which is analogous to our \cms{} problem while also encompassing minimization problems.  
They use a similar thresholding approach, but focus on bounding the ``commitment gap" of \cms, a quantity studied in several costly information settings. \robin{maybe point to some other commitment gaps/make this sentence better}
\citet{chawla2025combinatorial} does not consider prophet inequalities at all.

Finally, in subsequent work \citet{chawla2025commitment} extend the techniques of our paper's Section \ref{sec:inefficient} to achieve improved bounds on the commitment gap of the $\cms$ problem.

\subsubsection*{Acknowledgements}
This work was supported by the National Science Foundation award \#2329431.
The authors thank Bobby Kleinberg, Laura Doval, and participants at the University of Texas seminar for discussion and feedback.

\vfill
\break

\bibliographystyle{alphaurl}
\bibliography{citations}

\vfill
\break

\appendix

\section{Proofs for Section \ref{sec:inefficient}} \label{app:inefficient}

We prove Lemma \ref{lemma:alg-cobs-cms}, restated here for convenience.
\cobstocms*

\begin{proof}

Given a policy $\rho'$ on $\C_i$, construct a policy $\rho$ on $\M_i$ as follows.
For every deterministic policy $\pi$, let $\rho$ select $\pi$ with the same probability that $\rho'$ selects bandit $B_i^{\pi}$.
Whenever $\rho'$ advances from state $s$, take action $\pi(h)$ from the corresponding history $h$.
If $\rho'$ halts or claims on a state $s$, halt or claim, respectively, on the corresponding history $h$.
It follows from Definition \ref{def:bandit} that $\E[\Perf^{\pi}(\M_i)]$ = $\E[\Perf^{\rho'}(B^\pi_i)].$
Note that $\rho$ claims exactly when $\rho'$ respectively claims, so $\rho$ and $\rho'$ claim with the same probability.

Now, given a policy $\rho$ on $\M_i$, construct a policy $\rho'$ on $\C_i$ as follows.
Note that $\rho$ WLOG first draws a deterministic policy $\pi$ from some distribution and then follows $\pi$.
Construct $\rho'$ by drawing a deterministic policy $\pi$ according to $\rho$, selecting bandit $B_i^{\pi}$, advancing to a terminal state $s$, and claiming if $V(s)>0$.
If $V(s)=0$, claim if and only if $\pi$ claims on the corresponding history $h$.
By exactly the argument given above, $\E[\Perf^{\rho}(\M_i)]$ = $\E[\Perf^{\rho'}(B^\pi_i)]$ and $\rho'$ claims with the same probability as $\rho$.
\end{proof}

We now prove Proposition \ref{prop:cms-to-cobs}, restated here for convenience.
\propcmstocobs*

\begin{proof}
  Given an instance $\I$ of $\cms(\F)$, we construct $\I' = \MSPtoCOB(\I)$.
  Because $\ALG'$ is online, it executes an independent policy $\rho_i'$ on each $\C_i$ before either claiming $\C_i$ or proceeding to $C_{i+1}$.
  By Lemma \ref{lemma:alg-cobs-cms}, for each $\rho_i'$, there exists a policy $\rho_i$ on $\M_i$ such that $\E[\Perf^{\rho_i'}(\C_i)] = \E[\Perf^{\rho_i}(\M_i)]$.
  Let $\ALG$ be the online algorithm for $\cms(\F)$ that executes policy $\rho_i$ on MSP $\M_i$ for all $i\in[n]$.
  We have that
  \[
  \E[\Welf^{\ALG}(\I)] = \sum_{i=1}^n \E[\Perf^{\rho_i'}(\C_i)] = \sum_{i=1}^n \E[\Perf^{\rho_i}(\M_i)] = \E[\Welf^{\ALG'}(\I')].
  \]
  By Lemma \ref{lemma:alg-cobs-cms}, because $\ALG$ claims $\M_i$ exactly when $\ALG'$ claims $\C_i$, and $\ALG'$ is ex-post feasible, $\ALG$ is ex-post feasible.
  
  Recall that Lemma \ref{lemma:alg-cobs-cms} also tells us that, given a policy $\rho$ for $\M_i$, there exists a policy $\rho$ for $\C_i$ with the same performance and probability of claiming.
  By an identical argument to the above, given an algorithm $\ALG$ for $\cms(\F)$, there is a feasible algorithm $\ALG'$ for \cob{} Selection$(\F)$ with the same expected welfare.

  Now, let $\OPT$ and $\OPT'$ be the optimal ex-ante feasible algorithms for \cms{} and \cob{} Selection respectively.
  By Observation \ref{obs:ex-ante-indep}, each consists of independently executing policies $\rho_1,\ldots,\rho_n$ on $\M_1,\ldots,\M_n$ (in the case of $\OPT$)
  or $\rho_1',\ldots,\rho_n'$ on $\C_1, \ldots, \C_n$ (in the case of $\OPT'$.)
  Then, applying the above two arguments, there exists an ex-ante algorithm for $\cms(\F)$ with the same expected welfare as $\OPT'$, and there exists an ex-ante algorithm
  for \cob{} Selection$(\F)$ with the same expected welfare as $\OPT$.
  It follows that $\Welf^{OPT}(\I) = \Welf^{\OPT'}(\I')$.
  
  Given an online ex-ante $\alpha$-approximation algorithm $\ALG'$ for \cob{} Selection$(\F)$, we can construct an online algorithm $\ALG$ for $\cms(\F)$ with the same expected welfare.
  We have that
  \[
  \E[\Welf^{\ALG}(\I)] = \E[\Welf^{\ALG'}(\I')] \geq \alpha \cdot \E[\Welf^{\OPT'}(\I')] = \alpha \cdot \E[\Welf^{\OPT}(\I)],
  \]
  so $\ALG$ is an ex-ante $\alpha$-approximation algorithm for $\cms(\F)$.

\end{proof}

We construct a policy for \cob{} selection that is used in \ref{lemma:alg-cors-cobs}.

\begin{lemma}[non-exposed policy]\label{lemma:non-exposed-policy}
 For any bandit $B$ with capped value $\kappa$, given probability $q$, there exists a non-exposed policy $\rho$ on $B$ which claims with probability exactly $q$,
 and also claims exactly when the value of $\kappa$ is the top $q$ quantile of its distribution.
\end{lemma}

We construct $\rho$ to advance from state $s$ whenever $\sigma_s$ is in the top $q$ quantile of the distribution of $\kappa$.
Because our probability distributions are discrete, describing $\rho$ precisely requires a tie-breaking procedure.
We set a threshold $\tau$ based on the distribution of $\kappa$, and then choose, possibly randomly, between the following two policies:
either advance from the current state $s$ if and only if $\sigma_s > \tau$, or if and only if $\sigma_s \geq \tau$.
We call policies that, for some $\tau$, either advance if and only if $\sigma_s > \tau$ or advance if and only if $\sigma_s \geq \tau$ \emph{threshold-based}.
The concept of threshold-based algorithms is explored in detail in Section \ref{sec:efficient}.
We first prove that all threshold-based policies are non-exposed.

\begin{lemma}\label{lemma:non-exposed-threshold}
All threshold-based policies are non-exposed.
\end{lemma}

\begin{proof}
Suppose for contradiction that $\rho$ is exposed.
Then there is positive probability that $\rho$ advances from a state $s$ and later halts in a state $s'$ such that $\sigma_s < \sigma_{s'}$.
But since $\rho$ advanced from state $s$, we either have that $\tau \leq \sigma_s < \sigma_{s'}$, or $\tau < \sigma_s < \sigma_{s'}$.
In either case, $\tau < \sigma_{s'}$ and $\rho$ should have advanced from $s'$, a contradiction. 
\end{proof}

We now prove Lemma \ref{lemma:non-exposed-policy}. 
\begin{proof}[Proof of Lemma \ref{lemma:non-exposed-policy}]

  If there exists $v\in \text{Supp}(\kappa)$ such that $\Pr[\kappa\geq v] = q$, set a threshold $\tau = v$ and advance from state $s$ whenever $\sigma_s \geq \tau$.
  With probability exactly $q$, $\kappa = \min_s\{\sigma_s\} \geq \tau$, so $\sigma_s \geq \tau$ for all $\sigma_s$, and therefore we claim with probability exactly $q$.

  Otherwise, there is some $v\in \text{Supp}(\kappa)$ such that $\Pr[\kappa \geq v] > q \geq \Pr[\kappa > v]$\footnote{the second inequality is strict except when $q=0$, in which case $v$ is the maximum element of Supp$(\kappa)$ and $\Pr[\kappa > v]=0$.}
  Set $\tau=v$.
  With probability $\frac{q-\Pr[\kappa > \tau]}{\Pr[\kappa=\tau]}$, follow the threshold policy of advancing whenever $\sigma_s \geq \tau$.
  Otherwise, follow the threshold policy of advancing whenever $\sigma_s > \tau$.
  $\rho$ claims whenever $\kappa > \tau$, and when $\kappa=\tau$, claims with probability $\frac{q-\Pr[\kappa > \tau]}{\Pr[\kappa=\tau]}$.
  Therefore, the probability that $\rho$ claims is
  \[\Pr[\kappa > \tau] + \Pr[\kappa=\tau] \cdot \frac{q-\Pr[\kappa>\tau]}{\Pr[\kappa=\tau]} = \Pr[\kappa > \tau] + (q - \Pr[\kappa > \tau]) = q.\]

  $\rho$ claims with probability $q$, and moreover claims exactly when $\kappa$ is in its top $q$ quantile.
  By Lemma \ref{lemma:non-exposed-threshold}, $\rho$ is non-exposed, as it randomizes between threshold-based policies.

\end{proof}

\section{Proofs for Section \ref{sec:efficient}} \label{app:efficient}

\subsection{Proof of Proposition \ref{prop:ex-ante-approx} (Ex-ante FPTAS)} \label{subsec:ex-ante-approx-pf}

In this section, we give an algorithm to compute an $\epsilon$-approximate collection.
We recall the definition and the statement of Proposition \ref{prop:ex-ante-approx}.

\defepsapproxcoll*

\propexanteapprox*

We first recall the background and define some notation.
Recall that $\F$ is a matroid and the input to the problem are MSPs $\M_1,\dots,\M_n$.
We assume WLOG that $\{i\} \in \F$ for all $i$, as otherwise we can discard $i$ from the input.
Recall from Observation \ref{obs:ex-ante-indep} that the optimal ex-ante feasible policy, without loss of generality, commits to a statistically independent policy $\rho_i$ on each $\M_i$.

Given an ex-ante feasible algorithm whose output is a set $S$, we define the following functions for each $i \in [n]$: letting $\rho$ range over all possibly-randomized policies for $\M_i$,
\begin{align*}
  f_i(q) := & ~ \max_{\rho} ~ \E[\Perfofon{\rho}{i}]  \\
            & \quad \text{s.t.} \quad \Pr_{\rho}[i \in S] \leq q .
\end{align*}
The optimal ex-ante feasible algorithm solves $\max_{Q \in \PF} \sum_{i=1}^n f_i(Q_i)$.

Our main ingredient is the following algorithm.
\begin{proposition} \label{prop:ex-ante-MSP}
  There is an algorithm that, on input $c,\epsilon > 0$ and a MSP $\M_i$, runs in time polynomial in the description length of $\M$, in $\log\tfrac{1}{c}$, and in $\tfrac{1}{\epsilon}$, and produces an explicitly-represented, monotone, concave function $\hat{f}_i: [0,1] \to \reals$ such that, for all $q \in [0,1]$,
  \[ \frac{f_i(q)-c}{1+\epsilon} \leq \hat{f}_i(q) \leq (f_i(q) + c)(1+\epsilon) . \]
\end{proposition}

We now take Proposition \ref{prop:ex-ante-MSP} as given and prove Proposition \ref{prop:ex-ante-approx}.

\begin{algorithm}[\textsc{approx-collection}] \phantom{.} \bigskip \\
\textsc{approx-collection}$(\M_1,\dots,\M_n,\epsilon)$:
\begin{enumerate}
  \item Set $\epsilon' := \tfrac{1}{3} \min\left\{\epsilon, \tfrac{1}{2}\right\}$.
  \item Compute the expected performance of \SAUP$(\M_i,0)$ for each $i$, and let $W$ be the largest.
  \item Set the parameter $c := \frac{\epsilon' W}{2n}$.
  \item Run the algorithm of Proposition \ref{prop:ex-ante-MSP} on each $M_i$ to obtain $\hat{f}_i$.
  \item Solve the convex programming problem $Q := \arg\max_{Q \in \PF} \sum_{i=1}^n \hat{f}_i(Q_i)$.
  \item For each $i$, let $\hat{U}_i := \max\left\{0, ~ \frac{\hat{f}_i(Q_i)}{1+\epsilon} - c \right\}$.
  \item Return $(Q,\hat{U})$.
\end{enumerate}
\end{algorithm}

\begin{proof}[Proof of Proposition \ref{prop:ex-ante-approx}]
  \emph{(Running time.)}
  Computing the expected performance of \SAUP$(\M_i,0)$ takes polynomial time, as the expected performance is computed in the course of constructing the \SAUP{} oracle in Proposition \ref{prop:saup}.
  The algorithm of Proposition \ref{prop:ex-ante-MSP} is called $n$ times and runs in time polynomial in the length of the given MSP, in $\tfrac{1}{\epsilon}$, and in $\log\tfrac{1}{c} = \log \tfrac{2n}{\epsilon' W}$.
  The convex programming problem is solvable in polynomial time because it is a maximization of a concave function over a convex polytope, that, because it is a matroid polytope, has an efficient separation oracle~\citep{cunningham1984testing}.
  We note that, in addition to explicit representations of $\hat{f}_i$, we also have access to the gradients and supergradients of $\hat{f}_i$, a piecewise linear function (see proof of Proposition \ref{prop:ex-ante-MSP}).

  \emph{($\epsilon$-approximate collection.)}
  Let $\OPT^Q$ be the optimal ex-ante feasible algorithm that claims each $i$ with probability $Q_i$.
  First, we note that for any feasible $Q' \in \PF$, we have
  \begin{align}
    \OPT^Q
    &=    \sum_{i=1}^n f_i(Q) \nonumber \\
    &\geq \left(\frac{\sum_{i=1}^n \hat{f}_i(Q)}{1+\epsilon'}\right) - cn  & \text{Proposition \ref{prop:ex-ante-MSP}}  \nonumber \\
    &\geq \left(\frac{\sum_{i=1}^n \hat{f}_i(Q')}{1+\epsilon'}\right) - cn  & \text{definition of $Q$}  \nonumber \\
    &\geq \left(\frac{\sum_{i=1}^n f_i(Q')}{(1+\epsilon')^2}\right) - \frac{cn}{1+\epsilon'} - cn  & \text{Proposition \ref{prop:ex-ante-MSP}}  \nonumber \\
    &\geq (1-2\epsilon') \OPT^{Q'} - cn\frac{2+\epsilon'}{1+\epsilon'}   \nonumber \\
    &\geq (1-2\epsilon') \OPT^{Q'} - 2cn  \nonumber \\
    &\geq (1-2\epsilon') \OPT^{Q'} - \epsilon' W . \label{eqn:eps-appx-bound}
  \end{align}
  We first use (\ref{eqn:eps-appx-bound}) to derive that $\OPT^Q \geq W/2$.
  In particular, letting $\bar{Q}$ be the distribution putting probability $1$ on a particular $i$, we have $\OPT^{\bar{Q}} = \SAUP(\M_i,0) = W$.
  Therefore, 
  \begin{align*}
    \OPT^Q
    &\geq (1-2\epsilon') W - \epsilon' W  \\
    &=    (1-3\epsilon') W  \\
    &\geq \frac{W}{2},
  \end{align*}
  using that $\epsilon' \leq \tfrac{1}{6}$.
  We also have $\OPT \geq W$, as $W$ is the welfare of one feasible solution.

  We now show that $\OPT^Q$ is a $1-\epsilon$ approximation.
  Let $Q^*$ be the probabilities for the optimal algorithm, i.e. $\OPT = \OPT^{Q^*}$.
  By (\ref{eqn:eps-appx-bound}),
  \begin{align*}
    \OPT^Q
    &\geq (1-2\epsilon') \OPT - \epsilon' W  \\
    &\geq (1-3\epsilon') \OPT  \\
    &\geq (1-\epsilon) \OPT
  \end{align*}
  using that $\epsilon' \leq \epsilon/3$.
  Finally, we show that $\hat{U}$ approximates $U$ from below.
  First, by the guarantee of Proposition \ref{prop:ex-ante-MSP} and $U_i \geq 0$, we obtain $\hat{U}_i \leq f_i(Q_i) = U_i$ for all $i$.
  Second, we have
  \begin{align*}
    \sum_{i=1}^n \hat{U}_i
    &\geq \frac{\sum_{i=1}^n \hat{f}_i(Q_i)}{1+\epsilon'} - cn  & \text{definition of $\hat{U}$}  \\
    &\geq \frac{\sum_{i=1}^n f_i(Q_i)}{(1+\epsilon')^2} - \frac{cn}{1+\epsilon'} - cn  & \text{Proposition \ref{prop:ex-ante-MSP}}  \\
    &\geq (1-2\epsilon') \OPT^Q - 2cn  \\
    &=    (1-2\epsilon') \OPT^Q - \epsilon' W  \\
    &\geq (1-2\epsilon') \OPT^Q - \frac{\epsilon' \OPT^Q}{2}  \\
    &\geq (1-3\epsilon') \OPT^Q  \\
    &\geq (1-\epsilon) \OPT .
  \end{align*}
  We have shown the output is an $\epsilon$-approximate collection.
\end{proof}

\subsubsection{Proof of Proposition \ref{prop:ex-ante-MSP}}

Our algorithm will compute the following functions: \elias{Which algorithm? \textsc{approx-collection?}}
\begin{definition} \label{def:exante-dags-g}
  Given a MSP $\M = (S,A,P,C,V)$, define $g_s(q)$ and $g_{s,a}(q')$ as follows.
  At a terminal state $s$, let $g_s(q) = q V(s)$.
  At a non-terminal state $s$, let:
  \begin{align*}
    g_{s,a}(q') &= \max_{\{q_{s'}\}} \sum_{s'} P(s';a,s) g_{s'}(q_{s'})  & \text{s.t.} ~ \sum_{s'} P(s';a,s) q_{s'} = q'  \\
    g_s(q) &= \max_{\lambda,\{q_{s,a}\}} \sum_{a \in A^s} \lambda(a) \left[ g_{s,a}(q_{s,a}) - C(a,s) \right]  & \text{s.t.} ~ \sum_{a \in A^s} \lambda(a) q_{s,a} = q .
  \end{align*}
\end{definition}

For comparison, given any state $s$, let
 \[ f_s(q) = \max_{\bar{\pi}} \E[ \Perf^{\bar{\pi}}(s) ] ~~ \text{s.t.} ~~ \Pr[\text{$\bar{\pi}$ claims}] \leq q.  \]
\elias{Above equation was copy-pasted from EC draft, changed some notation but I think it may be quite wrong. $\Perf(s)$ isn't well-defined but whatever}
\begin{lemma} \label{lemma:exante-f-g}
  $f_s = g_s$.
\end{lemma}
\begin{proof}
  By backward induction.
  At a terminal state $s$, the only two possibilities are to claim reward $V(s)$ or not claim.
  The policy that claims with probability $q$ has expected reward $q \cdot V(s)$.

  At a non-terminal state $s$, we first argue that $g_{s,a}(q')$ is the optimal expected performance achievable by any policy that begins after action $a$ has been chosen in state $s$, i.e. begins by drawing $s' \sim P(\cdot;a,s)$.
  Suppose by backward induction that at all possible transition states $s'$, i.e. where $P(s';a,s) > 0$, we have $g_s(q) = f_s(q)$.
  The optimal policy $\bar{\pi}$ at $s$ that is required to take action $a$ in state $s$ and must claim with probability at most $q'$ has, for each possible $s'$, some probability $q_{s'}$ of claiming conditioned on transitioning to state $s'$.
  It must satisfy $\sum_{s'} P(s';a,s) q_{s'} = q'$.
  Furthermore, for each state $s'$, it runs the optimal policy $\bar{\pi}_{s'}$ starting at state $s'$ that claims with probability $q_{s'}$.
  Therefore, $g_{s,a}(q')$ achieves that optimum, i.e. maximizes expected performance over all policies.

  Now, we argue $g_s(q) = f_s(q)$, i.e., $g_s$ optimizes expected performance subject to claiming with probability $q$.
  The optimal such policy must first choose to either halt (action $\bot$) or take an action $a \in A^s$, according to some probability distribution $\lambda$.
  In doing so, it incurs a cost $C(a,s)$.
  For each action $a$, this optimal policy has some probability $q_{s,a}$ of claiming conditioned on taking action $a$.
  It must satisfy $\sum_{a \in A^s \cup \{\bot\}} \lambda(a) q_{s,a} = q$.
  And conditioned on choosing $a$, this optimal policy continues by running the policy that maximizes expected performance among all policies that begin after action $a$ has been chosen in state $s$ and that claim with probability $q_{s,a}$.
  That is, it conditioned on choosing $a$ and paying $C(a,s)$, the optimal policy obtains $g_{s,a}(q_{s,a})$.
  Therefore, $g_s(q)$ achieves that optimum.
\end{proof}

\begin{lemma} \label{lemma:exante-g-concave}
  $g_{s,a}$ and $g_s$ are concave and monotone nondecreasing.
  Furthermore, $g_{s,a}(0) = g_s(0)$ while $g_{s,a}(1)$ and $g_s(1)$ are computable in polynomial time.
\end{lemma}
\begin{proof}
  To show that $g_{s,a}(1)$ and $g_s(1)$ are computable in polynomial time, observe that the performance of an algorithm that never claims cannot be above zero, while an algorithm that always claims can be implemented by \SAUP{}$(\M,0)$.
  (see Proposition \ref{prop:saup})

  Now, we again use backward induction.
  At a terminal state $s$, $g_s(q)$ is linear in $q$ and nondecreasing because $V(s) \geq 0$.
  Now consider a non-terminal state $s$.
  Suppose for each $s'$ reachable from $s$, $g_s$ is concave and monotone nondecreasing.
  Then $g_{s,a}$ is monotone nondecreasing because, if $q < q'$ and $\{q_{s'}\}$ is a valid solution to the maximization for $g_{s,a}(q)$, then there is a solution $\{q'_{s'}\}$ for $g_{s,a}(q')$ that satisfies $q'_{s'} \geq q_{s'}$ for all $s'$. The claim then follows by monotonicity of each $g_{s'}$.

  Next, we directly prove that $g_{s,a}(q)$ is concave for all $a$.
  Let $\alpha \in [0,1]$ and $q,q' \in [0,1]$.
  Let $\{q_{s'}\}$ and $\{q'_{s'}\}$ respectively solve the maximization objectives for $g_{s,a}(q)$ and $g_{s,a}(q')$.
  Then
  \begin{align}
    \alpha g_{s,a}(q) + (1-\alpha) g_{s,a}(q')
    &= \sum_{s'} P(s';a,s) \left[ \alpha g_{s'}(q_{s'}) + (1-\alpha) g_{s'}(q'_{s'}) \right]  \nonumber \\
    &\leq \sum_{s'} P(s';a,s) g_{s'}(\alpha q_{s'} + (1-\alpha) q'_{s'}) .  \label{eq:exante-gsa-concave}
  \end{align}
  And by taking a convex combination of the constraints, they satisfy
  \begin{align*}
    \alpha q + (1-\alpha)q'
    &= \alpha \sum_{s'} P(s';a,s) q_{s'} + (1-\alpha) \sum_{s'} P(s';a,s) q'_{s'}  \\
    &= \sum_{s'} P(s';a,s) \left(\alpha q_{s'} + (1-\alpha) q'_{s'} \right) .
  \end{align*}
  So the collection $\{\alpha q_{s'} + (1-\alpha) q'_{s'}\}$ is a feasible solution for $g_{s,a}(\alpha q + (1-\alpha)q')$.
  By (\ref{eq:exante-gsa-concave}), that collection gives a weakly higher value than $\alpha g_{s,a}(q) + (1-\alpha)g_{s,a}(q')$.
  As $g_{s,a}(\alpha q + (1-\alpha)q')$ is the maximum over all feasible collections, we have $g_{s,a}(\alpha q + (1-\alpha)q') \geq \alpha g_{s,a}(q) + (1-\alpha)g_{s,a}(q')$.
  
  Next, we consider $g_s(q)$.
  For each $a$, the function $q_{s,a} \mapsto g_{s,a}(q_{s,a}) - C(a,s)$ is a concave, monotone increasing function minus a constant, so it is still concave, monotone increasing.
  Now, $g_s$ is the concave closure of $\{g_{s,a} - C(a,s) : a \in A^s \cup \{\bot\}\}$.
  That is, it is the pointwise smallest concave function that lies weakly above every $g_{s,a}(q) - C(a,s)$ function.
  This implies that $g_s(q)$ is monotone increasing and concave.
\end{proof}

We have already shown that $f_i$ is concave:
\begin{lemma} \label{lemma:exante-dags-f-concave}
  Given a MSP $\M_i$, the function $f_i$ is concave.
\end{lemma}
\begin{proof}
  We have that $f_i = f_{s^*}$, where $s^*$ is the start state of the MSP.
  By Lemma \ref{lemma:exante-f-g}, $f_{s^*} = g_{s^*}$, and by Lemma \ref{lemma:exante-g-concave}, $g_{s^*}$ is concave.
\end{proof}

Next, we give an efficient approximation guarantee for $f_{s^*} = g_{s^*}$, proving Proposition \ref{prop:ex-ante-MSP}.

\begin{algorithm}[\textsc{approx-}$g$] \label{alg:approx-g} \phantom{.} \bigskip
\begin{itemize}
  \item Given: $\M = (S,A,P,C,V)$, $\epsilon > 0$. Let $m = |S|$. Define $\alpha = \epsilon/(4m)$.
  \item Parameter: $c < 1$.
  \item Output: for each $s$ and each $a \in A^s$, concave functions $\hat{g}_s, \hat{g}_{s,a} : [0,1] \to \reals$, each implemented as an efficient oracle for values and gradients.
  \item Let $b = \frac{c}{\max_s V(s)}$.
  \item Let $\Theta = (0, b, b(1+\alpha), b(1+\alpha)^2, \dots, 1 )$. 
        Alternatively, notate the set $\Theta = (\theta_0,\dots,\theta_T)$.
  \item By backward induction on the DAG structure:
  \begin{itemize}
    \item At a terminal state $s$, we set $\hat{g}_s(q) = q V(s)$.
    \item At a non-terminal state $s$, for each action $a$, we first compute $\hat{g}_{s,a}(q)$ for each $q \in P$ by solving its concave optimization problem of Definition \ref{def:exante-dags-g}, using as oracles $\hat{g}_{s'}$.
          We then ``iron'' $\hat{g}_{s,a}$ as follows.
          We set $\hat{g}_{s,a}(b) = \min\{\hat{g}_{s,a}(b), b \max_s V(s)\}$.
          We then enforce monotonicity by setting $\hat{g}_{s,a}(\theta_t) = \max\{\hat{g}_{s,a}(\theta_t), \hat{g}_{s,a}(\theta_{t-1})\}$ for $t=2,\dots,T$.
          We then define $\hat{g}_{s,a}(q)$ to be the concave envelope of $\hat{g}$, which we can compute and represent in closed form efficiently by taking the convex hull of the values computed above.
    \item Then, at the same non-terminal state $s$, we first compute $\hat{g}_s(q)$ for each $q \in P$ by solving its own concave optimization problem of Definition \ref{def:exante-dags-g}, using as oracles $\hat{g}_{s,a}$, and ironing in the same way.
  \end{itemize}
\end{itemize}
\end{algorithm}

\begin{lemma} \label{lemma:exante-dags-approx-gs}
  For all $s$ and $a$, for all $q \in [0,1]$, we have $\frac{g_{s,a}(q) - c}{1+\epsilon} \leq \hat{g}_{s,a}(q) \leq (g_{s,a}(q) + c)(1+\epsilon)$ and $\frac{g_{s}(q) - c}{1+\epsilon} \leq \hat{g}_{s}(q) \leq (g_{s}(q) + c)(1+\epsilon)$.
\end{lemma}
\begin{proof}
  We use backward induction.
  If $s$ is a terminal state, $g_s$ is at level $0$ of the induction.
  Any $g_{s,a}$ where the maximum level is $k-1$ over all $s'$ with $P(s';a,s) > 0$, is at level $k$.
  Similarly, $g_s$ where the maximum level is $k-1$ over all $g_{s,a}$ is at level $k$.
  We use backward induction to prove that, at the $k$th level of induction, the following hold:
  \begin{align}
    q \in \Theta \quad &\implies \quad \frac{g_{s,a}(q) - c}{(1+\alpha)^{k-1}} \leq \hat{g}_{s,a}(q) \leq (\hat{g}_{s,a}(q) + c)(1+\alpha)^{k-1}  \label{eq:exante-dags-gsa-appx-1} \\
    q \not\in\Theta \quad &\implies \quad \frac{g_{s,a}(q) - c}{(1+\alpha)^{k}} \leq \hat{g}_{s,a}(q) \leq (\hat{g}_{s,a}(q) + c)(1+\alpha)^{k}   \label{eq:exante-dags-gsa-appx-2}\\
    q \in \Theta \quad &\implies \quad \frac{g_{s}(q) - c}{(1+\alpha)^{k-1}} \leq \hat{g}_{s}(q) \leq (\hat{g}_{s}(q) + c)(1+\alpha)^{k-1}  \\
    q \not\in\Theta \quad &\implies \quad \frac{g_{s}(q) - c}{(1+\alpha)^{k}} \leq \hat{g}_{s}(q) \leq (\hat{g}_{s}(q) + c)(1+\alpha)^{k} .
  \end{align}
  Base case: at a terminal state $s$, $\hat{g}_s = g_s$, i.e. zero error.

  Now consider a non-terminal state at level $k$ of the induction.
  Consider $g_{s,a}(q)$ (Definition \ref{def:exante-dags-g}).
  First fix any $q \in \Theta$ and consider the values initially computed by the algorithm before ironing.
  Let $h_{s,a}(\{q_{s'}\}) = \sum_{s'} P(s';a,s) g_{s'}(q_{s'})$.
  Consider $\hat{h}_{s,a}(\{q_{s'}\}) = \sum_{s'} P(s';a,s) \hat{g}_{s'}(q_{s'})$.
  Because it is a convex combination, we have
    \[ \frac{h_{s,a}(\{q_{s'}\}) + c}{(1+\alpha)^{k-1}} \leq \hat{h}_{s,a}(\{q_{s'}\}) \leq (h_{s,a}(\{q_{s'}\}) + c)(1+\alpha)^{k-1} . \]
  This implies the same for $\hat{g}_{s,a}(q)$, i.e. (\ref{eq:exante-dags-gsa-appx-1}).
  \bo{could explain more.}

  Now, consider values computed at $q \in \Theta$ after enforcing monotonicity and taking the concave envelope.
  If $q=0$, then $\hat{g}_{s,a}(0) = g_{s,a}(0) = 0$.
  If $q=b$, then we either do not change $\hat{g}_{s,a}(b)$ or decrease it to equal $b \max_s V(s)$.
  As this is the maximum possible value of $g_{s,a}(b)$, we have not worsened the approximation guarantee of (\ref{eq:exante-dags-gsa-appx-1}).

  For other $q \in \Theta$, we maintain the guarantee (\ref{eq:exante-dags-gsa-appx-1}), as follows.
  In case 1, $\hat{g}_{s,a}(q)$ has not changed and (\ref{eq:exante-dags-gsa-appx-1}) holds.
  In case 2, $\hat{g}_{s,a}(q)$ has been set equal to $\hat{g}_{s,a}(q/(1+\alpha))$ to maintain monotonicity.
  In this case, since the value has been raised, it still lies above $g_{s,a}(q)(1-\alpha)$, and we still have $\hat{g}_{s,a}(q) \leq \hat{g}_{s,a}(q/(1+\alpha)) \leq g_{s,a}(q/(1+\alpha))(1+\alpha)^{k-1}  \leq g_{s,a}(q)(1+\alpha)^{k-1}$ by induction on $q \in \Theta$.
  In case 3, $\hat{g}_{s,a}(q)$ has changed to be larger so that it equals the convex combination of $\hat{g}_{s,a}$ at some $q',q'' \in \Theta$.
  Then the approximation applies to $\hat{g}_{s,a}(q'), \hat{g}_{s,a}(q'')$ and we obtain (\ref{eq:exante-dags-gsa-appx-1}) as well.

  Now fix any $q \in [b,1], q \not\in \Theta$.
  Let $x < q < y$ for $x,y \in \Theta$ with $y = x(1+\alpha)$.
  Note that because $g_{s,a}$ is concave and monotone, and because $g_{s,a}(0) = 0$, we have $dg_{s,a}(q) \geq g_{s,a}(q)/q$, where $dg_{s,a}$ is any supergradient.
  This implies $g_{s,a}(y) \leq g_{s,a}(x)(y/x) = g_{s,a}(x)(1+\alpha)$.
  By monotonicity, $g_{s,a}(q) \leq g_{s,a}(y) \leq g_{s,a}(x)(1+\alpha) \leq \hat{g}_{s,a}(x)(1+\alpha)^k \leq \hat{g}_{s,a}(q)(1+\alpha)^k$.
  Similarly, $g_{s,a}(q) \geq g_{s,a}(x) \geq g_{s,a}(y)/(1+\alpha) \geq \hat{g}_{s,a}(y)/(1+\alpha)^k \geq \hat{g}_{s,a}(q)/(1+\alpha)^k$.

  Finally, consider any $q \in (0,b)$.
  We claim that $g_{s,a}$ and $g_s$ are $\max_s V(s)$ Lipschitz.
  The result follows because each $g_{s,a}$ cannot have a steeper slope than one of the $g_{s'}$ than it depends on, and so forth, and the steepest slope at any terminal state is $\max_s V(s)$.
  It therefore holds that $g_{s,a}(q) \in [0, b \max_s V(s)]$, and we also have [$\hat{g}_{s,a}(q) \in 0, b \max_s V(s)]$.
  So (\ref{eq:exante-dags-gsa-appx-2}) holds, even if we remove the $(1+\alpha)^k$ term.

  We now repeat the argument at $\hat{g}_s$.
  Although the optimization problem is slightly different, the arguments are all exactly the same.
  In particular, we just utilize $h_{s}(\lambda,\{q_{s,a}\}) = \sum_a \lambda(a) g_{s,a}(q_{s,a})$ and $\hat{h}_s$ in the exactly analogous way, and then the entire argument is the same.
  This completes the induction.
  
  Finally, as $|S|=m$, there can be at most $2m$ levels of induction, so using that $(1+\alpha)^{2m} \leq e^{2\alpha m} = e^{\epsilon/2} \leq (1+\epsilon)$ completes the proof.
\end{proof}

\begin{proof}[Proof of Proposition \ref{prop:ex-ante-MSP}]
  Lemma \ref{lemma:exante-dags-approx-gs} implies the accuracy result.
  Considering the running time of the algorithm, for each state and action, we build a piecewise linear concave function $\hat{g}_{s,a}$ with $T+2$ pieces, where $T$ is the floor of the solution to $b(1+\alpha)^T = 1$, i.e. $T = O(\tfrac{1}{\alpha}\ln(\tfrac{1}{b})) = O(\tfrac{|S|}{\epsilon}\ln(\tfrac{\max_s V(s)}{b}))$.
  For each piece of the piecewise function, we solve a concave optimization problem subject to a linear constraint, using constant-time oracles for concave functions (including access to supergradients) that have already been constructed, and do polynomial-time post-processing \elias{too vague?}.
  The construction of $\hat{g}_s$ is efficient by a similar argument.
  We thus build a polynomial number of functions, each taking polynomial time.
\end{proof}

\subsubsection{The FPTAS}

We now show that the analysis above yields a Fully-Polynomial Time Approximation Scheme (FPTAS) for the optimal ex-ante feasible algorithm for CMS$(\F)$, for any matroid $\F$.
(In fact, the result extends to any downward-closed constraint for whose feasibility polytope there exists an efficient separation oracle.)

\begin{proposition} \label{prop:fptas-itself}
  There is an algorithm running in time polynomial in the lengths of $\M_1,\dots,\M_n$ and in $\frac{1}{\epsilon}$ that, for any $\epsilon > 0$, guarantees a $1-\epsilon$ approximation to the optimal ex-ante feasible algorithm for CMS$(\F)$.
\end{proposition}
\begin{proof}[Sketch]
  Our algorithms compute a vector $Q \in \PF$ such that, if we optimally interact with each $\M_i$ subject to claiming with probability $Q_i$, then we achieve a $1-\epsilon$ approximation to the problem $\max_Q \sum_{i=1}^n \hat{f}_i(Q)$.
  Within each $\M_i$, at the source state $s$, we compute a probability distribution $\lambda_s$ over available actions and a probability $q_{s,a}$ for each action such that \emph{(a)} $\sum_a \lambda_s(a) q_{s,a} = Q_i$, and \emph{(b)} it is optimal to draw an action $a$ from $\lambda_s$, then proceed optimally subject to claiming with probability $q_{s,a}$.
  Our algorithm recursively computes such a distribution at every intermediate state, so it implements the optimal algorithm for the given $Q$ and approximations $\hat{f}_i$.
  We lose another factor of $1-\epsilon$ because $\hat{f}_i$ only approximates $f_i$.
  \bo{There's something orphaned in the previous submission about kappa interpretations for ex-ante opt.}
\end{proof}

\subsection{Proof of Proposition \ref{prop:matroid-prophet}} \label{subsec:matroid-prophet-pf}

In this section, we prove:

\propmatroidprophet*

We will reduce to the simpler Choice-Over-Rewards Selection setting.
Recall that in this setting, each alternative $i$ is in the form of a ``cabinet'' with some number of closed ``drawers''.
Each drawer $j$ contains a reward $X_i^j$, and only one drawer may be opened.
The rewards in the drawers may be correlated.
We summarize the input as the collection $(X_i^j)$, representing a list of random variables ranging over all $i \in [n]$ and each ``drawer'' $j$ of each $i$.

Our algorithm will interact with the instance through two oracles.
The first is a \SAUP{} oracle, as we will construct threshold-based algorithm.
The second is an $\epsilon$-approximate collection $(Q,\hat{U})$ that describes an approximate solution to the ex-ante optimum.
We now note how these definitions apply in the special case of the \cor{} Selection setting.

\begin{definition}[Threshold-based, \cor{} Selection] \label{def:cor-threshold-based}
  In the \cor{} Selection setting, an online algorithm is threshold-based (Definition \ref{def:saup}) if it follows the following format for each arrival $i$:
  \begin{itemize}
    \item Set a threshold $\tau_i$.
    \item Solve the Single-Agent Utility Problem on $i$ with $\tau_i$:
    \begin{itemize}
      \item Pick the drawer $j$ maximizing $\E[(X_i^j - \tau)^+]$, open it, and claim $i$ if and only if $X_i^j \geq \tau$.
    \end{itemize}
  \end{itemize}
\end{definition}

The definition of an $\epsilon$-approximate collection is also the same, but we state it in the \cor{} Selection setting for convenience.
\begin{definition}[$\epsilon$-approximate collection, \cor{} Selection] \label{def:cor-eps-approx-coll}
  Given an input $(X_i^j)$ to the ex-ante \cor{} Selection$(\F)$ problem, an $\epsilon$-approximate collection (Definition \ref{def:eps-approx-coll}) is a pair of vectors $(Q,\hat{U})$ satisfying the following.
  First, $Q \in \PF$, the feasible polytope.
  Second, there exists an ex-ante feasible algorithm $\OPTHAT$ that is a $1-\epsilon$ approximation and satisfies $Q_i \geq \Pr[i \in \OPTHAT]$ for all $i$.
  Third, letting $U_i$ be the expected utility $\OPTHAT$ obtains on $i$, $\hat{U}$ approximates $U$ from below, i.e. $\hat{U}_i \leq U_i$ and $\sum_{i=1}^n \hat{U}_i \geq (1-\epsilon)\sum_{i=1}^n U_i$.
\end{definition}

We now give our ``matroid prophet'' algorithm for the \cor{} Selection setting.
The input is the instance $(X_i^j)$, an $\epsilon$-approximate collection $(Q,\hat{U})$ along with an oracle for \SAUP{}.
The following algorithm and analysis contains many ideas from existing matroid prophet algorithms such as those contained in \citet{kleinberg2012matroid,lee2018optimal}, but with some subtle changes as well.

First, some notation.
We will define a set of random variables $Z = (Z_1,\dots,Z_n)$.
For that set, define for any $S \in \F$:
  \[ R(S) := \max_{S'} \left\{\sum_{i \in S'} Z_i  ~ \middle| ~ S' \cap S = \emptyset, S' \cup S \in \F \right\} . \]
$R(S)$ may be interpreted as the optimal total value that can be added to $S$ while remaining feasible, with value measured by $Z$.

\begin{algorithm}[\mp{}] \label{alg:matroid-prophet}
The \mp{} algorithm takes an $\epsilon$-approximate collection $(Q,\hat{U})$ for some $\epsilon \geq 0$ and an oracle for the $\SAUP$ problem.
The algorithm is defined in both the $\cms$ and $\cors$ settings because it only interacts with the arrivals $\M_1,\dots,\M_n$ through the given $\SAUP$ oracle.

\bigskip
$\mp((X_i^j), (Q,\hat{U}), \SAUP(\cdot, \cdot))$:
\begin{itemize}
  \item Define $\hat{z}_i := \frac{\hat{U}_i}{Q_i}$.
  \item Compute a distribution $D_Q$ on $2^{[n]}$ satisfying $Q = \E_{F \sim D_Q} 1_F$, where $1_F$ is the indicator vector for $i \in F$.
  \item Define $Z$ as follows: sample $F \sim D_Q$, then set $Z_i = \hat{z}_i \Indic{i \in F}$.
  \item Initialize $A_0 = \emptyset$, and for each arrival $i=1,\dots,n$:
  \begin{itemize}
    \item If $A_{i-1} \cup \{i\} \not\in \F$, the set $\tau_i = \infty$, discard $i$, set $A_i = A_{i-1}$, and continue to $i+1$.
    \item Otherwise, set $\tau_i = \tfrac{1}{2} \E_{Z} \left[ R(A_{i-1}) - R(A_{i-1} - A_i) \right]$.
    \item Run \SAUP{}$(i, \tau_i)$. If it claims, set $A_i = A_{i-1} \cup \{i\}$, else set $A_i = A_{i-1}$.
  \end{itemize}
\end{itemize}
\end{algorithm}

\begin{lemma} \label{lemma:matroid-prophet-efficient}
  If $\F$ is a matroid, then the algorithm runs in polynomial time.
\end{lemma}
\begin{proof}
  We can Computed $D_Q$ in polynomial time with convex programming, obtaining an explicit representation of $D_Q$ (i.e. a probability distribution as a vector of probabilities, necessarily of at most polynomial size).
  This is because there is an efficient separation oracle for any matroid polytope~\citep{cunningham1984testing}, and by using such an oracle, one can obtain $\D_Q$ explicitly in polynomial time for any polytope (e.g. Theorem 4 of \citet{cai2017constructive}, citing older work).

  Other than solving the \SAUP{} problem, for which we have an oracle, the only other nontrivial step is computing $\tau_i$.
  Because $D_Q$ is explicitly represented, we simply sum over the support of $D_Q$ and enumerate all realizations of $F \sim D_Q$, which determines all possible realizations of the vector $Z$, in order to compute the sum.
\end{proof}

We now turn to analyzing the approximation guarantee.
For this, we first introduce more notation.
Given the instance $(X_i^j)$ and the vector $Q \in \PF$, define $U_i$ to be the utility obtained on arrival $i$ by $\OPT^Q$, the optimal algorithm that claims each $i$ with probability at most $Q_i$.
Define $z_i = \frac{U_i}{Q_i}$.
Recall that by definition of an $\epsilon$-approximate collection, we have $\hat{U}_i \leq U_i$ for all $i$.

The algorithm $\OPT^Q$, WLOG, opens a drawer $j^*(i) \sim \lambda$ for each arrival $i$ according to some fixed probability distribution $\lambda$, independently of all other arrivals.
Let $X_i := X_i^{j^*(i)}$, a random variable representing the reward\footnote{We note that the remainder of the analysis from this point would be much more cumbersome in the CMS setting as opposed to the \cor{} Selection setting that we are considering.} observed by $\OPT^Q$ on arrival $i$.

Reconsider the random variables $F$ and $Z$ defined in the algorithm.
For the purpose of analysis, we additional define the ``re-randomized'' random variables $Y = (Y_1,\dots,Y_n)$ as follows: If $i \in F$, then draw $Y_i$ from the distribution of $X_i$ conditioned on being in its top $Q_i$ quantile; otherwise, draw $Y_i$ from the distribution of $X_i$ conditioned on not being in the top $Q_i$ quantile. 

We note that $Y_i$ has the same marginal distribution as $X_i$, and that $\E[Y_i \mid i \in F] = \frac{U_i}{Q_i} = z_i$.

For the purpose of analysis, we also consider an independent copy $(F',Z')$ of $(F,Z)$ drawn from the same distribution.
Define $R'(S)$ analogously to $R(S)$, but for $Z'$.

\begin{lemma} \label{lemma:matroid-prophet-approx}
  For any matroid $\F$, \mp{} run on an $\epsilon$-approximate collection $(Q,\hat{U})$ with a \SAUP{} oracle has an expected performance of at least $\frac{1}{2} \sum_{i=1}^n \hat{U}_i$.
\end{lemma}
\begin{proof}
  For shorthand, write $X := (X_i^j)$.
  Let $j(i)$ denote the drawer of arrival $i$ opened by our algorithm.
  Let $\text{Util}_i(X) = (X_i^{j(i)} - \tau_i)^+$, interpreted as the utility of agent $i$ who opens drawer $j(i)$ and pays $\tau_i$ for the reward iff it exceeds $\tau_i$.
  Let $\text{Util}(X) = \sum_{i=1}^n \text{Util}_i(X)$.
  Note that because $i \in A_n \iff X_i^{j(i)} \geq \tau_i$, where $A_n$ is the algorithm's claimed set, we also have $\text{Util}(X) = \sum_{i \in A_n} \text{Util}_i(X)$.
  Therefore, for any fixed realization of $X$,
  \begin{align}
    \text{Welf}^{\text{ALG}}
    &= \left(\sum_{i \in A_n} \tau_i\right) ~ + ~ \text{Util}(X)  \nonumber \\
    &= \frac{1}{2} \E_Z \sum_{i \in A_n} \left( R(A_{i-1}) - R(A_{i-1} \cup \{i\}) \right) ~ + ~ \text{Util}(X)  \nonumber \\
    &= \frac{1}{2} \E_Z \sum_{i \in A_n} \left( R(A_{i-1}) - R(A_{i}) \right) ~ + ~ \text{Util}(X)  \nonumber \\
    &= \frac{1}{2} \E_Z \left[ R(A_{0}) - R(A_{n}) \right] ~ + ~ \text{Util}(X) .  \label{eqn:mp-breakdown}
  \end{align}
  We first lower-bound the second term, $\text{Util}(X)$.
  First, we use the definition of \SAUP{} to observe that $\text{Util}_i(X) \geq \E \left(X_i - \tau_i\right)^+$.
  This uses independence across arrivals and holds for all fixed realizations of all prior arrivals $i'=1,\dots,i-1$ and all choices of $\tau_i$, so it holds in expectation.
  
  Next, we use the fact that $X_i$ is independent of $\tau_i$ and that $Y_i$ has the same marginal distribution to conclude that $\E (X_i - \tau_i)^+ = \E (Y_i - \tau_i)^+$.
  Therefore, we have $\text{Util}_i(X) \geq \E (Y_i - \tau_i)^+$.

  Now, given realizations of $X,F,Z$, define $V = \arg\max_{S'} \left\{ \sum_{i \in S'} Z_i \middle| S' \cap A_n = \emptyset, S' \cup A_n \in \F \right\}$.
  In other words, $V$ is the set whose total weight is $R(A_n)$.
  We will use that, WLOG, $i \in V$ iff $Z_i = \hat{z}_i$, as otherwise $Z_i = 0$.
  \begin{align}
    \E \text{Util}(X)
    &\geq \E \sum_{i=1}^n (Y_i - \tau_i)^+  \nonumber \\
    &\geq \E \sum_{i \in V} (Y_i - \tau_i)^+  \nonumber \\
    &\geq \E \sum_{i \in V} (Y_i - \tau_i)  \nonumber \\
    &=    \E \sum_{i \in V} Y_i ~ - ~ \E \sum_{i \in V} \tau_i  \nonumber \\
    &=    \E \sum_{i \in V} \E[ Y_i \mid i \in F] ~ - ~ \E \sum_{i \in V} \tau_i  \nonumber \\
    &=    \E \sum_{i \in V} z_i ~ - ~ \E \sum_{i \in V} \tau_i  \nonumber \\
    &\geq \E \sum_{i \in V} \hat{z}_i ~ - ~ \E \sum_{i \in V} \tau_i  \nonumber \\
    &=    \E \sum_{i \in V} Z_i ~ - ~ \E \sum_{i \in V} \tau_i  \nonumber \\
    &=    \E R(A_n) ~ - ~ \E \sum_{i \in V} \tau_i  .  \label{eqn:mp-util} \\
  \end{align}

  Plugging (\ref{eqn:mp-util}) into (\ref{eqn:mp-breakdown}), we obtain
  \begin{align}
    \text{Welf}^{\text{ALG}}
    &\geq \frac{1}{2} \E R(\emptyset) + \frac{1}{2} \E R(A_n) - \E \sum_{i \in V} \tau_i . \label{eqn:mp-halfway}
  \end{align}

  We now claim the following.
  \begin{lemma} \label{lemma:kw12}
    $\E \sum_{i \in V} \tau_i \leq \frac{1}{2} \E R(A_n)$.
  \end{lemma}
  \begin{proof}
    Let $(F',Z')$ be drawn as an independent copy of $(F,Z)$, and let and $R'$ be the corresponding marginal gain function. Then
    \begin{align*}
      \E \sum_{i \in V} \tau_i = \E_{X,F} \sum_{i \in V} \frac{1}{2} \E_{F'} \left[ R'(A_{i-1}) - R'(A_{i-1} \cup \{i\}) \right] .
   \end{align*}
    Fix any realizations of $X,F,F'$.
    Proposition 2 of \citet{kleinberg2012matroid} gives the following: for every matroid $\F$, for every list of $n$ numbers $(z_1,\dots,z_n)$, and for every pair of disjoint sets $A,V$ with $A \cup V \in \F$, letting $A_i = A \cap [i]$, defining $F^* = \max_{S \in \F} \sum_{i \in S} z_i$ and $R'(S) = \max_{S' \subseteq F^*} \left\{ \sum_{i \in S'} z_i ~:~ S' \cap A = \emptyset, S' \cup A \in \F \right\}$, we have
    \begin{align*}
       \sum_{i \in V} \left[ R'(A_{i-1}) - R'(A_{i-1} \cup \{i\}) \right] &\leq R(A_n) .  & \text{Proposition 2 of \citet{kleinberg2012matroid}}
    \end{align*}
    Applying this to our $V, Z', R', A$, and taking expectations,
    \begin{align*}
      \frac{1}{2} \E \sum_{i \in V} \left[ R'(A_{i-1}) - R'(A_{i-1} \cup \{i\}) \right]
     &\leq \frac{1}{2} \E R'(A_n) .
  \end{align*}
  \end{proof}
  Plugging Lemma \ref{lemma:kw12} into (\ref{eqn:mp-halfway}), we obtain:
  \begin{align*}
    \text{Welf}^{\text{ALG}}
    &\geq \frac{1}{2} \E R(\emptyset) + \frac{1}{2} \E R(A_n) - \frac{1}{2} \E R(A_n) \\
    &= \frac{1}{2} \E R(\emptyset) \\
    &= \frac{1}{2} \E \sum_{i \in F} Z_i  \\
    &= \frac{1}{2} \sum_{i=1}^n Q_i \hat{z}_i  \\
    &= \frac{1}{2} \sum_{i=1}^n \hat{U}_i .
  \end{align*}
\end{proof}

\bo{ideally: here or below the next proof, we would state a proposition that we have a $1/2-\epsilon$ approx for CORS efficiently, and in fact for $\epsilon=0$ by adding in the lemma that we can compute ex-ante opt efficiently for CORS.}

We can now prove Proposition \ref{prop:matroid-prophet}, that Algorithm \ref{alg:matroid-prophet} (\mp), run on a \cms{} instance and $(Q,\hat{U})$, is efficient and provides an approximation guarantee relative to $\hat{U}$.
\begin{proof}[Proof of Proposition \ref{prop:matroid-prophet}]
  Efficiency is Lemma \ref{lemma:matroid-prophet-efficient}.

  Let $\I$ be an instance of $\cms$, let $\I' = \MSPtoCOB(\I)$, and let $\I'' = \COBtoCOR(\I')$.
  Let $(Q,\hat{U})$ be an $\epsilon$-approximate collection for $\I$.
  Then it is also an $\epsilon$-approximate collection for $\I''$, by the reductions of Lemmas \ref{lemma:alg-cobs-cms} and \ref{lemma:alg-cors-cobs}. \bo{more detail would help}

  Given $(Q,\hat{U})$, let $\ALG$ be the $\cms$ version of $\mp$ run on $\I$, while $\ALG''$ is the $\cors$ version of $\mp$ run on $\I''$.
  Then $\E \Welfofon{\ALG}{\I} = \E \Welfofon{\ALG''}{\I''}$, because for each arrival $i$ and threshold $\tau_i$, $\SAUP(i,\tau_i)$ has the same expected performance and probability of claiming.
  \bo{following line is for arxiv only, not camera-ready}
  \bo{This whole proof could benefit from a rewrite for journal version}
  In fact, both obtain performance $\max_{\pi} \E[(\kappa_i^{\pi} - \tau_i)^+]$, where $\kappa_i^{\pi}$ is the capped value of the bandit that mimics following policy $\pi$ on $\M_i$, and the maximum is over all deterministic policies; see Proposition \ref{prop:weiztman-stuff} for details.
  \bo{end: arxiv-only}
  Lemma \ref{lemma:matroid-prophet-approx} proves $\E \Welfofon{\ALG''}{\I''} \geq \tfrac{1}{2}\sum_{i=1}^n \hat{U}_i$, proving the claim.
\end{proof}

\section{Hardness of \cms{}} \label{app:hardness}

\elias{TODO: example of non-obligatory as CMS, maybe including figure, other problems that we can express as CMS (redo appendix title accordingly)}

The problem of Pandora's Non-Obligatory Inspection, introduced by \citet{doval2018whether}, can be naturally expressed as a special case of \cms.
In the \emph{Pandora's Non-Obligatory Inspection} (\pnoi{}) problem, an algorithm is presented with $n$ boxes $\{1,\ldots,n\}$.
Box $i$ contains a value $v_i$ distributed according to a distribution $D_i$, and has an associated inspection cost $c_i$ to open.
At each step, the algorithm may open some box $i$, paying its cost $c_i$ and revealing its value $v_i$, drawn independently from $D_i$,
claim the value of a box that has previously been opened and halt, or claim the value in an unopened box $i$ without paying the inspection cost $c_i$ and halt.
We denote the welfare of an algorithm $\ALG$ for an instance $\I$ of \pnoi{} by $\Welf^{\ALG}(\I)$, in keeping with our other notation.
Any \pnoi{} instance in which every $D_i$ has finite support can be represented as a special case of \cms, with one small subtlety regarding claiming an unopened box. 

Let $\F$ be the rank-one matroid constraint $\mathcal{F}=\{\emptyset,\{1\},\{2\},\ldots,\{n\}\}$. That is, we may only select a single item.
Given an instance $\I$ of \pnoi{}, we construct an instance $\I'$ of $\cms(\F)$ as follows.

For each box $i$, construct a MSP $\M_i$ with two legal actions from the start state $s_i^*$, ``inspect'' and ``take''.
Create terminal states $t_i$ and $s_i^1, \dots, s_i^{d_i}$ for each of the $d_i$ possible values in $\text{Supp}(D_i)$. 
The ``inspect'' action has cost $c_i$ and transitions to a terminal state $s_i^k$ with probability $\Pr[D_i=v_i^k]$, for every $v_i^k\in\text{Supp}(D_i)$.
The ``take'' action deterministically transitions to state $t_i$ with value $\E[D_i]$ and incurs no cost.
Note that because the ``take'' action is free and its transition function is deterministic, WLOG, any algorithm that chooses action ``take'' on some $\M_i$ immediately claims $\M_i$.

\citet{fu2023pandora} showed that computing the optimal policy for \pnoi{} is NP-hard, and specifically that computing optimal policies on a subclass of ``low-cost-low-return-support-3" (LCLRS3) problems is NP-hard. 
Crucially, this set of instances only uses boxes with three possible realizations, allowing construction of a \cms{} instance of size polynomial in the number of boxes of the \pnoi{} instance. 
It follows that \cms{} is NP-hard.
\elias{make more formal? Theorem/proof environment? Alternatively make \emph{less} formal especially if we list other problems}

\section{Efficiency of the Inefficient Reduction} \label{app:bonus}

In Section \ref{sec:inefficient}, we showed that any ex-ante prophet inequality translates to an online approximation algorithm for $\cms(\F)$.
However, the latter is not generally computationally efficient.
In this section, we provide more detail on which steps of the reduction can be made efficient generically.

As summarized in Figure \ref{fig:results-overview}, we use a sequence of reductions:
\begin{itemize}
  \item In Proposition \ref{prop:cms-to-cobs}, we reduce $\cms(\F)$ to Choice-Over-Bandits Selection ($\cobs(\F)$). In general, this reduction is not efficient because the size of the instance is increased exponentially: for each $\msp$ with a DAG state graph, we create a choice-over-bandits process with one choice for each path in the original DAG. \bo{However, if it's a tree, then I guess everything is computationally efficient?}
  \item In Proposition \ref{prop:cobs-to-cors}, we reduce $\cobs(\F)$ to Choice-Over-Rewards Selection ($\cors(\F)$). It is almost immediate that this reduction is efficient (as shown below).
  \item In Proposition \ref{prop:cors-to-prophets}, we reduce $\cors(\F)$ to a classic ex-ante prophet inequality with constraint $\F$. Nontrivially, this reduction turns out to be efficient for certain classes of $\F$.
\end{itemize}

\bo{Update 2026-03-06: explained monotone below.}
Therefore, we will show the following result.
Here a prophet inequality algorithm is \emph{monotone}\footnote{For any non-monotone prophet inequality, there exists a monotone one with the same guarantee. We do not know of any cases where a non-monotone but computationally efficient algorithm is known, yet a monotone computationally efficient algorithm with the same performance is unknown. It seems unlikely, but we have not ruled out the possibility for all feasibility structures.} if, when it observes a value and then claims it, it would also have claimed any strictly larger value; see e.g. \citet{kleinberg2012matroid}.
\begin{proposition}
  If the polytope $\PF$ admits a polynomial-time separation oracle, and there is a polynomial-time monotone ex-ante $\alpha$ prophet inequality for constraint $\F$, then there is a polynomial-time online $\alpha$-approximation algorithm for Choice-Over-Bandits Selection ($\cobs(\F)$).
\end{proposition}
\begin{proof}
  Suppose we have a monotone polynomial-time ex-ante $\alpha$ prophet inequality.
  
  First, in Proposition \ref{prop:cors-prophets-efficient} below, we show that there is a polynomial-time algorithm for computing the ex-ante optimal algorithm for $\cors(\F)$ and in particular the associated probability distribution $\lambda_i$ from which to pick a drawer on each arriving ``cabinet'' $i$.
  We now observe that, given $(\lambda_1,\dots,\lambda_n)$, our reduction in Proposition \ref{prop:cors-to-prophets} is immediately computationally efficient, as it simply opens a drawer according to $\lambda_i$ and feeds the observed value to the prophet-inequality algorithm.
  This completes the efficient alogrithm for $\cors(\F)$.

  \bo{Do we need monotonicity?}
  We now produce an algorithm for $\cobs(\F)$.
  First, following Observation \ref{obs:quantile-based}, we can make the monotone algorithm for $\cors(\F)$ quantile-based.
  We can furthermore do so efficiently: we simply compute the probability $q_i^j$ with which $i$ is claimed conditioned on drawer $j$ being opened, and modify the algorithm to claim iff $X_i^j$ is in its top $q_i^j$ quantile (see Observation \ref{obs:quantile-based}). \bo{possibly tricky point: these are not the same $q_i^j$ as in the ex-ante opt algorithm below; these come from the prophet inequality algorithm}

  Now, for $\cobs(\F)$, we simply check that the reduction of Proposition \ref{prop:cobs-to-cors} is already efficient.\footnote{The reduction requires the $\cors(\F)$ algorithm to be quantile-based, which we have ensured above.}
  When $i$ arrives, we draw $j \sim \lambda_i$ computed by the $\cors(\F)$ algorithm above.
  Then, we simply claim bandit process $B_i^j$ if its capped Weitzman index $\kappa_i^j$ is in its top $q_i^j$ quantile.
  As discussed in Proposition \ref{prop:cobs-to-cors}, this is achievable by simply advancing the bandit until the Weitzman index drops below a fixed threshold (possibly with random tiebreaking), claiming iff the index remains above.
\end{proof}
  
We now complete this section by stating and proving Proposition \ref{prop:cors-prophets-efficient}.
In particular, we compute $\Lambda=(\lambda_1,\ldots,\lambda_n)$, the set of distributions from which ex-ante $\OPT$ select a ``drawer'' from each ``cabinet''.
This turns out to be a convex programming problem over the polytope $\PF$.
Therefore, we can efficiently calculate $\Lambda$ and the ex-ante $\OPT$ policy for \cors{} for downward-closed constraints such as matroids where there exists an efficient separation oracle.

\begin{proposition} \label{prop:cors-prophets-efficient}
  Let $\F$ be downward-closed such that there exists a polynomial-time separation oracle for $\PF$.
  Given an instance $\I$, let $\lambda_i$ be the distribution from which the ex-ante optimal algorithm $\OPT$ for $\I$ selects a drawer $j^*(i)$ on \cor{} $\C_i$.
  There is a polynomial-time algorithm that, given $\I$, computes the ex-ante $\OPT$ algorithm and in particular computes $(\lambda_1, \ldots, \lambda_n)$.
\end{proposition}

\begin{proof}
By Observation \ref{obs:ex-ante-indep}, ex-ante $\OPT$ WLOG executes an independent policy $\rho_i$ on each $\C_i$.
By Observation \ref{obs:quantile-based}, each $\rho_i$ is WLOG quantile-based, selecting $X_i^j$ if and only if its realization is in the top $q_i^j$ quantile
of its distribution for some $q_i^j\in[0,1]$.
Define $g_i^j(q) := q \E[ X_i^j \mid \text{$X_i^j$ in top $q$ quantile} ]$, the ``upper expectation'' of $X_i^j$ for quantile $q$.
We first claim that $g_i^j$ is a concave function (this fact is also asserted in \citet{feldman2016online}, Section 4.1).
To see this, let $G(q) = x$ such that\footnote{If there are multiple $x$ where this is true, they lie in a region where $X_i^j$ has zero density and any choice of $x$ in this region is valid for the argument.} $\Pr[X_i^j \leq x] = 1-q$.
In other words, $G(q) = F^{-1}(1-q)$ where $F$ is the CDF.
We have $g_i^j(q) = \int_{y=q}^1 dG(y)$; this is a version of the fact $\E[X_i^j] = \int_0^{\infty} (1 - F(x)) dx$.
And $G$ is a positive, monotone decreasing function, so (noting the bounds on the integral) $g_i^j$ is concave.

Now, given a fixed $Q_i$, the optimal algorithm on cabinet $i$ that claims it independently with probability $Q_i$ solves
\[ f_i(Q_i) = \max_{\lambda_i \in \Delta_{[m_i]}, (q_i^j : j \in [m_i])} \sum_{j \in [m_i]} \lambda_i(j) g_i^j(q_i^j) , \]
subject to $\sum_{j \in [m_i]} \lambda_i(j) q_i^j = Q_i$.
We can rephrase.
Let $\beta_i \in \Delta_{[0,1]}$ be a distribution over points $q \in [0,1]$ satisfying that $\E_{q \sim \beta_i} q = Q_i$.
Define
\[ \hat{f}_i(Q_i) = \max_{\beta_i \in \Delta_{[0,1]}} \E_{q \sim \beta_i} \max_{j \in [m_i]} g_i^j(q) . \]
We claim that $\hat{f}_i = f_i$.
We have $\hat{f}_i(q) \geq f_i(q)$ because any valid solution $(\lambda_i, \{q_i^j\})$ can be rephrased as a distribution $\beta_i$.
But because each $g_i^j$ is concave, WLOG, $\beta_i$ is never supported on multiple points where the same $g_i^j$ achieves the maximum; it is better by Jensen's inequality to put all their weight on the average of the points.
So the solution $\beta_i$ is WLOG a distribution over $m_i$ points $\{q_i^j\}$ where $j = \argmax_{j'} g_i^{j'}(q_i^j)$, so $f_i$ can achieve the same value as $\hat{f}_i$.
In fact, $\hat{f}_i$ is just the concave envelope of $(g_i^j : j \in [m_i])$, i.e. the pointwise smallest concave function that lies weakly above every $g_i^j$.

Given $Q_i$, $\hat{f}_i(Q_i)$ is efficiently computable: it is a concave maximization problem over the linear constraint that the expectation of $\beta_i$ is $Q_i$.
The solution $\beta_i$ produces $\lambda_i, \{q_i^j\}$.
We also note that the gradient of $f_i$ at $Q_i$ is computable efficiently as well.

$\hat{f}_i(Q_i)$ is also decreasing, as it represents the maximum value obtainable by selecting with probability $Q_i$.
This implies that the function $h(Q_i) := Q_i f_i(Q_i)$ is concave, as we can make the same argument as with $g_i^j$.
The ex-ante relaxation  problem is
\[ \max_{Q \in \PF} \sum_i Q_i f_i(Q_i) , \]
a concave maximization problem subject to a polytope constraint.
If $\PF$ has an efficient separation oracle, it is efficiently solvable to find $Q$.
Given $Q$, we have already shown that $\lambda_i$ is efficiently computable.
\end{proof}

\section{Weitzman Index Interpretation of \SAUP{}}\label{sec:weitzman-interpretation}

The efficient \SAUP{} algorithm described in Proposition \ref{prop:saup} can be nicely characterized in the language of the Weitzman indices we use for bandit processes, discussed in Section \ref{subsec:bandits}.
We first develop some notation.
Given a MSP $\M$ and a state $s$, consider the \msp{} $\M'$ that we get by taking $s$ to be our start state.
Given a deterministic policy $\pi$ on $\M'$, let $B^{\pi}$ be the bandit process induced by $\pi$.
Define $\sigma_s^{\pi}$ to be the Weitzman index $\sigma_{s^*}$ of the start state $s^*$ of $B^{\pi}$,
and $\kappa_s^{\pi}$ to be $\kappa_{s^*}$.

We prove that from any starting state $s$ in $\M$, the value achieved by the algorithm in Proposition \ref{prop:saup}, is exactly equal to a Weitzman amortization of the costs.

\begin{proposition} \label{prop:weiztman-stuff}
For any $\M'$ rooted at state $s$, 
$\E[\Perf^{\pi^*}(\M') - \A^{\pi^*}\tau] = \E[(\kappa_s^{\pi^*}-\tau)^+] = \max_\pi \E[(\kappa_s^{\pi}-\tau)^+]$. 
\end{proposition}

\begin{proof}
\robin{changed}
We will prove via induction that $\Val^{\pi^*,s} \max_\pi \E[(\kappa_s^{\pi}-\tau)^+]$. 
For $B^{\pi^*}$, the bandit induced by $\pi^*$, it immediately follows that $\Val^{\pi^*,s} = \E[(\kappa_s^{\pi^*}-\tau)^+]$. 
We will finally argue that $\E[\Perf^{\pi^*}(\M') - \A^{\pi^*}\tau] = \E[(\kappa_s^{\pi^*}-\tau)^+]$, completing the proof. 
\robin{end change} 

We proceed with the inductive portion. 
For a terminal state $s$, $\Val^{\pi^*,s}=(V(s)-\tau)^+ = \max_\pi \E[(\kappa_s^{\pi}-\tau)^+]$.
For a nonterminal state $s$, assume as our inductive hypothesis that $\Val^{\pi^*,n(s)} = \max_\pi \E[(\kappa_{n(s)}^{\pi}-\tau)^+]$ for all possible successors $n(s)$.
By definition of our algorithm, 
\begin{align*}
	\Val^{\pi^*,s} &= \max_a v^{\pi^*}(a; s)\\
	&= \max_a \E_{s' \sim p_{a,s}}[\Val^{\pi^*,s'}] - C(a,s)\\
	&= \max_a \max_\pi \E_{s' \sim p_{a,s}}[(\kappa_{n(s)}^{\pi}-\tau)^+] - C(a,s)
\end{align*}
by the inductive hypothesis. 

We define $\sigma_{a; s}^{\pi}$ as the Weitzman index for the bandit process starting at $s$ and taking action $a$ and then following policy $\pi$, referring to the Bandit process notation of Section \ref{subsec:bandits}.
Recall by the definition of the index $\sigma$ that $\E[(\kappa_{n(s)}^{\pi}-\sigma_{a; s}^{\pi})^+] = C(a, s)$ for any policy $\pi$. 
Noting that $C(a,s)$ is a constant for any action $a$, we can rewrite $\Val^{\pi^*,s}$ replacing the cost term with a constant function of the policy $\pi$:
\begin{align*}
	\Val^{\pi^*,s} &= \max_a \max_\pi \left(\E_{s' \sim p_{a,s}}[(\kappa_{s'}^{\pi}-\tau)^+] - \E[(\kappa_{s'}^{\pi}-\sigma_{a; s}^{\pi})^+]\right),
\end{align*}
with the maximum over $\pi$ now being taken over both terms in the expression.
We rewrite with $\pi^*$ and $a^*$ the maximizing policy and action chosen as $s$, respectively, to get
\begin{align}
	\Val^{\pi^*,s} &= \E_{s' \sim p_{a,s}}[(\kappa_{s'}^{\pi^*}-\tau)^+ - (\kappa_{s'}^{\pi^*}-\sigma_{a^*; s}^{\pi^*})^+]. \label{eq:val-pandora}
\end{align}

We consider two cases here: $\tau \geq \sigma_{a^*; s}^{\pi^*}$ and $\tau < \sigma_{a^*; s}^{\pi^*}$. 
In the first case, in the expectation of equation (\ref{eq:val-pandora}) the term $(\kappa_{s'}^{\pi^*}-\tau)^+$ will always be at most $(\kappa_{s'}^{\pi^*}-\sigma_{a^*; s}^{\pi^*})^+$, giving $\Val^{\pi^*,s}\leq 0$. 
By construction of the algorithm, if the best action would produce $\Val^{\pi^*,s} < 0$ it simply halts, giving $\Val^{\pi^*,s} = 0 = \Val^{\pi^*,s} = \E_{s' \sim p_{a,s}}[(\min(\sigma_{a^*; s}^{\pi^*}, \kappa_{s'}^{\pi^*})-\tau)^+] = \E_{s' \sim p_{a,s}}[(\kappa_{s}^{\pi^*}-\tau)^+]$. 

When $\tau < \sigma_{a; s}^{\pi^*}$, we analyze the interior of the expectation pointwise. 
Consider the realization of $\kappa_{s'}^{\pi^*}$ relative to $\sigma_{a^*; s}^{\pi^*}$. 
If $\kappa_{s'}^{\pi^*} < \sigma_{a^*; s}^{\pi^*}$, then 
\begin{align*}
(\kappa_{s'}^{\pi^*}-\tau)^+ - (\kappa_{s'}^{\pi^*}-\sigma_{a^*; s}^{\pi^*})^+ &= (\kappa_{s'}^{\pi^*}-\tau)^+\\
&= (\kappa_{s}^{\pi^*}-\tau)^+, 
\end{align*} 
with the last equality following the definition of $\kappa_{s}^{\pi^*} = \min(\sigma_{a^*; s}^{\pi^*}, \kappa_{s'}^{\pi^*})$. 
Otherwise, $\tau < \sigma_{a^*; s}^{\pi^*} \leq \kappa_{s'}^{\pi^*}$ and 
\begin{align*}
(\kappa_{s'}^{\pi^*}-\tau)^+ - (\kappa_{s'}^{\pi^*}-\sigma_{a^*; s}^{\pi^*})^+ &= \kappa_{s'}^{\pi^*}- \tau - \kappa_{s'}^{\pi^*} + \sigma_{a^*; s}^{\pi^*}\\
& = \sigma_{a^*; s}^{\pi^*} - \tau\\
& = (\sigma_{a^*; s}^{\pi^*} - \tau)^+ \\
& = (\kappa_{s}^{\pi^*} - \tau)^+,
\end{align*}
again following the definition of $\kappa_{s}^{\pi^*} = \min(\sigma_{a^*; s}^{\pi^*}, \kappa_{s'}^{\pi^*})$.
Therefore, when $\tau < \sigma_{a^*; s}^{\pi^*}$ we have $\Val^{\pi^*,s} = \E_{s' \sim p_{a,s}}[(\kappa_{s}^{\pi^*} - \tau)^+] = \max_{\pi} \E[(\kappa_{s}^{\pi} - \tau)^+]$ pointwise for all realizations.

\robin{changed}
We now argue that $\E[\Perf^{\pi^*}(\M') - \tau\A^{\pi^*}] = \E[(\kappa_s^{\pi^*}-\tau)^+]$ by claiming that $\pi^*$ is \emph{non-exposed} (Definition \ref{def:exposed}).
In the structure of the proof above, we consider two cases for $\tau$ vs. $\sigma_{a^*; s}^{\pi^*}$ at any state $s$. 
When $\sigma_{s}^{\pi^*}\leq \tau$, the policy halts.
Otherwise, the policy continues. 
This means that $\pi^*$ stops in state $s$ the first time $\sigma_{s}^{\pi^*}\leq \tau$, and in all previous states $s'$, $\sigma_{s'}^{\pi^*} > \sigma_{s}^{\pi^*}$. 
The policy is therefore non-exposed, and by Lemma \ref{lem:magic}, $\E[\Perf^{\pi^*}(\M')] = \E[\Perf^{\pi^*}(B^{\pi^*})] = \E[\A^{\pi^*}\kappa_{s}^{\pi^*}]$. 
Thus, $\E[\Perf^{\pi^*}(\M') - \tau\A^{\pi^*}] = \E[\A^{\pi^*}\kappa_s^{\pi^*} - \A^{\pi^*}\tau] = \E[\A^{\pi^*}(\kappa_s^{\pi^*} - \tau)] = \E[(\kappa_s^{\pi^*} - \tau)^+]$, as $\pi^*$ only claims if every intermediate $\sigma_{s'}^{\pi^*}$ is at least $\tau$.
\robin{end change}
\end{proof}

We highlight one fact used within the above proof: the policy $\pi^*$ is non-exposed.

\end{document}